\documentclass[12pt,a4paper,reqno]{amsart}
\allowdisplaybreaks

\usepackage[a4paper,left=3.5cm,right=3.5cm]{geometry}
\usepackage{amsaddr}
\usepackage{amsthm}
\usepackage{amssymb}
\usepackage{amsmath}
\usepackage{mathtools}
\usepackage{amsfonts}
\usepackage{bbm}
\usepackage{xcolor}
\usepackage[rightcaption]{sidecap}
\usepackage{enumitem}
\usepackage{mathabx}
\usepackage{tikz-cd}
\usepackage{caption}
\usepackage{ulem}	
\usepackage[abs]{overpic}

\usepackage{tikz}
\usetikzlibrary{intersections,fpu}

\usepackage{hyperref}
\hypersetup{
  colorlinks   = true, 
  urlcolor     = blue, 
  linkcolor    = blue, 
  citecolor   = blue, 
}

\newcommand{\rleq}{\trianglelefteq}
\newcommand{\rgeq}{\trianglerighteq}

\newtheorem{proposition}{Proposition}
\newtheorem{theorem}[proposition]{Theorem}
\newtheorem{lemma}[proposition]{Lemma}

\theoremstyle{definition}
\newtheorem{definition}[proposition]{Definition}
\newtheorem{example}[proposition]{Example}
\newtheorem{remark}[proposition]{Remark}

\definecolor{lighter}{RGB}{10.0,146.0,35.0}
\definecolor{darker}{RGB}{2.0,100.0,20.0}
\definecolor{sininen}{RGB}{89.5,101.8,238.1}
\definecolor{punainen}{RGB}{227.0,12.0,12.0}

\newcommand{\darker}[1]{\textcolor{darker}{#1}}
\newcommand{\punainen}[1]{\textcolor{punainen}{#1}}
\newcommand{\sininen}[1]{\textcolor{sininen}{#1}}

\newcommand{\mc}[1]{\mathcal{#1}}
\newcommand{\ms}[1]{\mathsf{#1}}
\newcommand{\mf}[1]{\mathfrak{#1}}
\newcommand{\mb}[1]{\mathbb{#1}}

\newcommand{\N}{\mathbb N}
\newcommand{\R}{\mathbb R}
\newcommand{\Z}{\mathbb Z}
\newcommand{\C}{\mathbb C}
\newcommand{\Q}{\mathbb Q}
\newcommand{\T}{\mathbb T}

\newcommand{\hil}{\mathcal{H}}

\newcommand{\tr}[1]{\mathrm{tr}\left[#1\right]} 
\def\<{\langle}
\def\>{\rangle}
\def\d{{\mathrm d}}
\newcommand{\id}{\mathbbm{1}} 
\newcommand{\fii}{\varphi}

\newcommand{\tj}{\vartheta}

\newcommand{\sis}[2]{\left\langle #1\middle| #2\right\rangle}

\let\originalleft\left
\let\originalright\right
\renewcommand{\left}{\mathopen{}\mathclose\bgroup\originalleft}
\renewcommand{\right}{\aftergroup\egroup\originalright}

\makeatletter
\newcommand*{\saved@uline}{}
\let\saved@uline\uline
\newcommand*{\mathuline}{%
  \mathpalette{\math@uline\saved@uline}%
}
\newcommand*{\math@uline}[3]{%
  \mbox{#1{$#2#3\m@th$}}%
}

\usepackage{todonotes}

\begin{document}

\title{Barycentric decompositions\\ for extensive monotone divergences}

\author{Erkka Haapasalo}
\address{Centre for Quantum  Technologies,  National University of Singapore}

\begin{abstract}
We study sets of divergences or dissimilarity measures in a generalized real-algebraic setting which includes the cases of classical and quantum multivariate divergences. We show that a special subset of divergences, the so-called test spectrum, characterizes the rest of the divergences through barycentres and that the extreme points of relevant convex subsets of general divergences are contained within the test spectrum. Only some special parts of the test spectrum may contain non-extreme elements. We are able to fully characterize the test spectrum in the case of classical multivariate divergences. The quantum case is much more varied, and we demonstrate that essentially all the bivariate and multivariate quantum divergences suggested previously in literature are within the test spectrum and extreme within the set of all quantum (multivariate) divergences. This suggests that the variability of quantum divergences is real since all the previously suggested divergences are independent of each other.
\end{abstract}

\maketitle

\section{Introduction}

Both in classical and quantum information theory we need to compare the information contained in a set of states (i.e.,\ probability measures in the classical case and density operators in the quantum case). In the classical case a set $P=\{p^{\tj}\}_{\tj\in\Theta}$ of probability measures, i.e.,\ a statistical experiment, is more informative than another statistical experiment $Q=\{q^\tj\}_{\tj\in\Theta}$ if $P$ attains a higher expected payoff in different tests that $Q$. According to Blackwell \cite{Blackwell51,Blackwell53}, this is equivalent with the existence of a stochastic operator (Markov kernel) $T$ such that $Tp^\tj=q^\tj$ for all $\tj\in\Theta$. In this case we denote $P\succeq Q$. We may study this comparison of experiments also in the quantum regime by defining that a quantum experiment $\vec{\rho}=\{\rho^\tj\}_{\tj\in\Theta}$ with states (density operators) $\rho^\tj$ is more informative than another experiment $\vec{\sigma}=\{\sigma^\tj\}_{\tj\in\Theta}$, symbolically $\vec{\rho}\succeq\vec{\sigma}$, if there is a quantum channel (a completely positive trace-preserving linear map) $\Phi$ such that $\Phi(\rho^\tj)=\sigma^\tj$ for all $\tj\in\Theta$.

We are often more interested in more relaxed settings of the comparison of experiments where we have access to multiple copies of the experiments or to suitable `catalytic' experiments. This means that we may ask when $\vec{\rho}=\big\{\rho^\tj\big\}_{\tj\in\Theta}$ yields more information than $\vec{\sigma}=\big\{\sigma^\tj\big\}_{\tj\in\Theta}$ {\it in large samples}, i.e.,
$$
\big\{\big(\rho^\tj\big)^{\otimes n}\big\}_{\tj\in\Theta}\succeq\big\{\big(\sigma^\tj\big)^{\otimes n}\big\}_{\tj\in\Theta}
$$
for any $n\in\N$ sufficiently large. On the other hand, we say that $\vec{\rho}$ yields more information than $\vec{\sigma}$ {\it catalytically} if there is an experiment $\vec{\tau}=\big\{\tau^\tj\big\}_{\tj\in\Theta}$ such that
$$
\big\{\rho^\tj\otimes\tau^\tj\big\}_{\tj\in\Theta}\succeq\big\{\sigma^\tj\otimes\tau^\tj\big\}_{\tj\in\Theta}.
$$
Some catalysts are essentially universial, i.e.,\ they can catalyze the information content comparison order for any inputs and outputs. This is the case, e.g.,\ for a fully orthogonal catalyst where, for any $\tj,\tj'\in\Theta$, $\tj\neq\tj'$, $\tau^\tj$ and $\tau^{\tj'}$ are orthogonally supported. This is why such trivial catalytic experiments are typically excluded.

We may attempt to approach the problem of comparing the information content in classical and quantum experiments (in the large-sample or catalytic regime) by using particular quantifiers which we call {\it divergences}: A map $D$ on quantum experiments is a divergence if $\vec{\rho}\succeq\vec{\sigma}$ implies $D(\vec{\rho})\geq D(\vec{\sigma})$ (data-processing inequality) and, for any experiments $\vec{\rho}=\{\rho^\tj\}_{\tj\in\Theta}$ and $\vec{\sigma}=\{\sigma^\tj\}_{\tj\in\Theta}$,
$$
D\left(\{\rho^\tj\otimes\sigma^\tj\}_{\tj\in\Theta}\right)=D(\vec{\rho})+D(\vec{\sigma})\quad\textrm{(tensor-additivity or extensivity)}.
$$
The data-processing inequality is a natural requirement for a measure of information content since such a measure should reflect the information ordering $\succeq$ of experiments. Extensivity means that the information content is additive under non-correlated products of experiments. The insistence on extensivity comes from the large-sample and catalytic comparison of experiments: ultimately we want to characterize large-sample or asymptotic majorization through inequalities $D(\vec{\rho})>D(\vec{\sigma})$ which are obviously unchanged in the large-sample and catalytic regimes if $D$ is extensive.

Let us recall that, if $\Theta=\{1,2\}$, divergences as defined above are often also called {\it relative entropies} which is arguably the term of choice in this particular setting where the term `divergence' is often used for any monotone map vanishing on repeating pairs $(\rho,\rho)$. This means that also certain one-shot quantities, e.g.,\ smoothed hypothesis-testing quantities, are often called divergences although they are not extensive. However, for us, all the divergences are henceforth assumed to be extensive.

In a generalized setting which includes both the classical and quantum information ordering of experiments, where the set $\Theta=\{1,\ldots,d\}$ is finite (and the sample spaces of the probability measures are finite and the Hilbert spaces are finite dimensional), it turns out that there is a special subset $\hat{\mf D}$ within the set of divergences which gives a necessary condition for the large-sample ordering which, in the quantum case, has the following form: whenever $\vec{\rho}=\big(\rho^{(1)},\ldots,\rho^{(d)}\big)$ and $\vec{\sigma}=\big(\sigma^{(1)},\ldots,\sigma^{(d)}\big)$ are (finite) quantum experiments such that $\Delta(\vec{\rho})>\Delta(\vec{\sigma})$ for all $\Delta\in\hat{\mf D}$, then
$$
\left(\big(\rho^{(1)}\big)^{\otimes n},\ldots,\big(\rho^{(d)}\big)^{\otimes n}\right)\succeq\left(\big(\sigma^{(1)}\big)^{\otimes n},\ldots,\big(\sigma^{(d)}\big)^{\otimes n}\right)
$$
for any $n\in\N$ sufficiently large. This result is a consequence of a real-algebraic Vergleichsstellensatz derived in \cite{FritzII}. Related results have also been proven in \cite{FritzI}.

This condition for large-sample information ordering implies in a straightforward manner that, given finite quantum (or classical) experiments $\vec{\rho}$ and $\vec{\sigma}$, the condition $\Delta(\vec{\rho})\geq\Delta(\vec{\sigma})$ for all $\Delta\in\hat{\mf D}$ implies that  $D(\vec{\rho})\geq D(\vec{\sigma})$ for all divergences $D$. In the classical case, for $d=2$, this was used to prove that any classical relative entropy (i.e.,\ a bivariate extensive divergence) can be considered as a barycentre over $\hat{\mf D}$ in \cite{Mu_et_al_2021}. In this special case the set $\hat{\mf D}$ consists of the maps $(p,q)\mapsto D_\alpha(p\|q)$ and $(p,q)\mapsto D_\alpha(q\|p)$ for all $\alpha\in[1/2,\infty]$, where the maps $D_\alpha$ are the bivariate R\'{e}nyi relative entropies \cite{Renyi61}. Recall that, for probability measures $p$ and $q$ on a standard Borel measurable space with supports ${\rm supp}\,p$ and ${\rm supp}\,q$ and a dominating measure $\nu$ (e.g.,\ $\nu=p+q$), we have
\begin{equation}\label{eq:RenyiRelEntr}
D_\alpha(p\|q)=\left\{\begin{array}{ll}
-\log{q({\rm supp}\,p)},&\alpha=0,\\
\frac{1}{\alpha-1}\log{\int \left(\frac{dp}{d\nu}\right)^\alpha\left(\frac{dq}{d\nu}\right)^{1-\alpha}}\,d\nu,&\alpha\in(0,1)\\
&{\rm or}\ \alpha\in(1,\infty)\ {\rm and}\ p\ll q,\\
\int \frac{dp}{dq}\log{\frac{dp}{dq}}\,dq,&\alpha=1\ {\rm and}\ p\ll q,\\
\log{\left\{q-{\rm ess}\,\sup \frac{dp}{dq}\right\}},&\alpha=\infty\ {\rm and}\ p\ll q,\\
\infty&{\rm otherwise},
\end{array}\right.
\end{equation}
with the convention $\log{0}=-\infty$ and $p\ll q$ means that $p$ is absolutely continuous with respect to $q$. Recall that $D_1=:D_{\rm KL}$ is also called the {\it Kullback-Leibler relative entropy} and $D_\infty=:D_{\rm max}$ is called the {\it max-relative entropy}. Thus, according to \cite{Mu_et_al_2021}, for any relative entropy $D$, there are inner and outer regular finite measures $\mu_\pm:\mc B\big([1/2,\infty]\big)\to[0,\infty)$ such that
\begin{equation}\label{eq:RenyiBarycentre}
D(p\|q)=\int_{[1/2,\infty]}D_\alpha(p\|q)\,d\mu_+(\alpha)+\int_{[1/2,\infty]}D_\alpha(q\|p)\,d\mu_-(\alpha).
\end{equation}

In the present work we generalize this result, so that, as special cases, we may characterize multivariate ($d$-variate) classical and quantum divergences as barycentres over special sets $\hat{\mf D}$ of divergences. In the classical case, the set $\hat{\mf D}$ can be identified with the set of classical multivariate R\'{e}nyi divergences studied in \cite{Farooq_et_al_2024,Verhagen_et_al_2024} where also the above large-sample result was derived. In the quantum case (as well as in more general settings), the set $\hat{\mf D}$ consists of a few parts: a subset we denote $\mf D_{\rm nd}$ (`nd' for `nondegenerate' or `non-derivation') and subsets $\mf D_k$, $k=1,\ldots,d$, of `derivations'. We have a normalization function $N:\mf D_{\rm nd}\to\R\setminus\{0\}$ such that, for all $\Delta\in\mf D_{\rm nd}$, there is a positive function $\Phi_\Delta$ on non-normalized experiments (i.e.,\ $d$-tuples of positive semi-definite operators of any non-zero trace) such that
\begin{equation}\label{eq:DeltaND}
\Delta(\vec{\rho})=N(\Delta)\log{\Phi_\Delta(\vec{\rho})}.
\end{equation}
The function $\Phi_\Delta$ is a monotone homomorphism in the sense that
\begin{itemize}
\item[(ND1)] we have either
\begin{align*}
\Phi_\Delta\left(\rho^{(1)}\oplus\sigma^{(1)},\ldots,\rho^{(d)}\oplus\sigma^{(d)}\right)&=\Phi_\Delta(\vec{\rho})+\Phi_\Delta(\vec{\sigma})\qquad {\rm or}\\
\Phi_\Delta\left(\rho^{(1)}\oplus\sigma^{(1)},\ldots,\rho^{(d)}\oplus\sigma^{(d)}\right)&=\max\left\{\Phi_\Delta(\vec{\rho}),\Phi_\Delta(\vec{\sigma})\right\}
\end{align*}
for all experiments $\vec{\rho}=\big(\rho^{(1)},\ldots,\rho^{(d)}\big)$ and $\vec{\sigma}=\big(\sigma^{(1)},\ldots,\sigma^{(d)}\big)$ which can be non-normalized,
\item[(ND2)] we have
$$
\Phi_\Delta\left(\rho^{(1)}\otimes\sigma^{(1)},\ldots,\rho^{(d)}\otimes\sigma^{(d)}\right)=\Phi_\Delta(\vec{\rho})\Phi_\Delta(\vec{\sigma})
$$
for all experiments $\vec{\rho}=\big(\rho^{(1)},\ldots,\rho^{(d)}\big)$ and $\vec{\sigma}=\big(\sigma^{(1)},\ldots,\sigma^{(d)}\big)$ which can be non-normalized,
\item[(ND3)] $\Phi_\Delta(1,\ldots,1)=1$ where we consider $1$ as a state on a one-dimensional Hilbert space, and
\item[(ND4)] whenever $\vec{\rho}\succeq\vec{\sigma}$, then $\Phi_\Delta(\vec{\rho})\geq\Phi_\Delta(\vec{\sigma})$ when $N(\Delta)>0$ and $\Phi_\Delta(\vec{\rho})\leq\Phi_\Delta(\vec{\sigma})$ when $N(\Delta)<0$.
\end{itemize}
The set $\mf D_k$ consists of divergences $\Delta$ which has the form $\Delta(\vec{\rho})=\Delta'(\vec{\rho})/\tr{\rho^{(k)}}$ for any non-normalized $\vec{\rho}=\big(\rho^{(1)},\ldots,\rho^{(d)}\big)$ where the map $\Delta'$ additionally satisfies
\begin{itemize}
\item[(k1)] $\Delta'\left(\rho^{(1)}\oplus\sigma^{(1)},\ldots,\rho^{(d)}\oplus\sigma^{(d)}\right)=\Delta'(\vec{\rho})+\Delta'(\vec{\sigma})$ and
\item[(k2)] $\Delta'\left(\rho^{(1)}\otimes\sigma^{(1)},\ldots,\rho^{(d)}\otimes\sigma^{(d)}\right)=\Delta'(\vec{\rho})\tr{\sigma^{(k)}}+\tr{\rho^{(k)}}\Delta'(\vec{\sigma})$
\end{itemize}
for all non-normalized experiments $\vec{\rho}=\big(\rho^{(1)},\ldots,\rho^{(d)}\big)$ and $\vec{\sigma}=\big(\sigma^{(1)},\ldots,\sigma^{(d)}\big)$.

The set $\mf D_{\rm nd}$ is further divided into four subsets depending which of the two instances in conditions (ND1) and (ND4) hold. Note that, when $\Delta\in\mf D_{\rm nd}$, conditions (ND2) -- (ND4) just mean that $\Delta$ is a divergence. Condition (ND1) is the special extra condition which characterizes $\mf D_{\rm nd}$ within the special set $\hat{\mf D}$. Conditions (k1) and (k2) are both extra conditions for any $\Delta\in\mf D_k$, $k=1,\ldots,d$. We may impose an additional normalization condition for all (quantum) divergences: Let us pick an experiment $\vec{\tau}=\big(\tau^{(1)},\ldots,\tau^{(d)}\big)$ of states $\tau^{(k)}$ such that $k\neq\ell$ implies $\tau^{(k)}\neq\tau^{(\ell)}$. Let us consider the convex set $\mf D$ of general (quantum) divergences $D$ such that $D(\vec{\tau})=1$. For the function $N:\mf D_{\rm nd}\to\R\setminus\{0\}$ it is now opportune to choose $N(\Delta)=\big(\log{\Phi_\Delta(\vec{\tau})}\big)^{-1}$ for all $\Delta\in\mf D_{\rm nd}$. We show that, in the general setting including the classical and quantum cases,
\begin{equation}\label{eq:extresult}
{\rm ext}\,\mf D\subseteq\hat{\mf D},\quad \mf D_{\rm nd}\subseteq{\rm ext}\,\mf D
\end{equation}
for the set ${\rm ext}\,\mf D$ of convex extreme points of the normalized set $\mf D$ of divergences. The latter inclusion above means that, in the quantum case, all the quantum divergences $\Delta$ having the form presented in \eqref{eq:DeltaND} where $\Phi_\Delta$ satisfies (ND1) -- (ND4) (especially the special property (ND1)) are extreme divergences. This means that (ND1) can be seen as an extremality condition for a divergence. Particularly all these special quantum divergences are independent of each other as none of them can be expressed as a barycentre with support outside a singleton. This prompts one to think that at least all $\Delta\in\mf D_{\rm nd}$ (and all $\Delta\in{\rm ext}\,\mf D_k$, $k=1,\ldots,d$) are relevant when characterizing the information order in quantum experiments. Since most of the proposed bivariate and multivariate quantum divergences are within $\mf D_{\rm nd}$ or $\mf D_k$ ($k=1,\ldots,d$) (as we shall demonstrate), this means that the large variability of quantum divergences presented in existing literature is non-redundant. Examples of $\Delta\in\mf D_{\rm nd}$ include the $\alpha$-$z$ quantum R\'{e}nyi divergences ($d=2$) \cite{Audenaert_Datta_2015,Jaksic2012} including the quantum divergences of the Petz-type \cite{Petz_85,Petz_1986} and of the `sandwiched' type \cite{Muller-Lennert_et_al_2013} as well as the Kubo-Ando divergences ($d=2$) \cite{Kubo_Ando_79} and the general multivariate divergences arising from matrix means \cite{Bhatia_Karandikar_2012}. Examples of $\Delta\in\mf D_k$ include the Umegaki quantum divergence and the Belavkin-Staszewski divergence \cite{Belavkin_Staszewski_82} ($d=2$ in both). Several multivariate generalizations of these were introduced, e.g.,\ in \cite{BuVra2021,mosonyi2024geometric}.

The theoretical background of this work lies in the theory of preordered semirings and asymptotic spectra. This theory has been highly useful especially for information theory ever since Volker Strassen found the first subcubic matrix multiplication algorithm, a find giving rise to his theory of asymptotic spectra \cite{strassen1986,strassen1987,strassen1991,strassen1998}. This theory has been applied to varied fields of information theory and computer science since then, even quantum information theory, and some extensions to Strassen's {\it Positivstellensatz} are represented by the results of \cite{FritzI,FritzII,Vrana2022}. An overview of the theory of asymptotic spectra and its uses is provided by \cite{WigZuid}. In this current work, we especially utilize a {\it Vergleichsstellensatz} (`comparison set theorem') from \cite{FritzII} which characterizes large sample ordering in a preordered semidomain by inequalities involving a (suitably compact) set of conditions (a `test spectrum') whose elements can be seen as special divergences and it coincides in the quantum and classical cases with the set $\hat{\mf D}$ of those $\Delta$ (or $\Phi_\Delta$) satisfying (ND1)--(ND4) or those $\Delta$ satisfying (k1) and (k2), $k=1,\ldots,d$. We apply similar machinery to this large-sample result as that used in \cite{Mu_et_al_2021} to derive the barycentric decompositions for general divergences. In the quantum case, this can be briefly stated as follows: for any (multivariate) quantum divergence $D$ there exists a finite positive (inner and outer regular) measure $\mu$ on the Borel $\sigma$-algebra of $\hat{\mf D}$ (w.r.t.\ to a particular compactifying topology of $\hat{\mf D}$) such that
$$
D(\vec{\rho})=\int_{\hat{\mf D}}\Delta(\vec{\rho})\,d\mu(\Delta)
$$
for all quantum experiments $\vec{\rho}$ consisting of trace-one density operators. Let us here underline the fact that although we know that the above barycentric expression exists for quantum divergences, we are currently unable to fully characterize the set $\hat{\mf D}$. This set however contains the majority of quantum divergences proposed in literature as we have just specified. The set $\hat{\mf D}$ for the classical multivariate case is characterized in detail in Example \ref{ex:classical} based on earlier results in \cite{Farooq_et_al_2024}.
The barycentric result is then used to prove the extremality result summed up in \eqref{eq:extresult}.

This paper is organized as follows: In Section \ref{sec:Background} we make some definitions used throughout this work and briefly review the machinery introduced in \cite{FritzI,FritzII} needed for the {\it Vergleichsstellensatz} establishing the large-sample comparison result laying the foundation on later results. Using the large-sample result and following the proof strategy presented in \cite{Mu_et_al_2021}, we derive a barycentric decomposition for any generalized divergence in Subsection \ref{subsec:barycentres}. We prove the (partial) characterization of extreme divergences presented in \eqref{eq:extresult} in Subsection \ref{subsec:ext}. In Section \ref{sec:quantum} we discuss the implications of these results to the quantum case, reviewing the divergences presented in existing literature and demonstrating their connections to the sets $\hat{\mf D}$ and $\mf D_{\rm nd}$. In Subsection \ref{subsec:quantumsemiring} we explicitly show that the quantum case is truly within the framework of the theory presented in \cite{FritzI,FritzII}, so that our results also apply to general multivariate quantum divergences.

\section{Background}\label{sec:Background}

We use the following basic notations throughout this work:
\begin{itemize}
\item $\N=\{1,2,3,\ldots\}$
\item $\R_{>0}$: the set of all strictly positive real numbers.
\item $\R_+$: The set of non-negative real numbers seen as a preordered semiring (see the definition of this notion later) with the natural addition, multiplication, and (total) order $\geq$.
\item $\R_+^{\rm op}$: as a set the same as $\R_+$ above but with the opposite order $a\geq^{\rm op}b$ $\Leftrightarrow$ $b\geq a$. $\R_+$ and $\R_+^{\rm op}$ together are often called {\it temperate non-negative reals}.
\item $\T\R_+$: otherwise the same as $\R_+$ above but with the tropical sum $a+'b=\max\{a,b\}$.
\item $\T\R_+^{\rm op}$: the same as $\T\R_+$ above but with the opposite order. $\T\R_+$ and $\T\R_+^{\rm op}$ together are often called as {\it tropical non-negative reals}.
\end{itemize}

For any $n\in\N$, we define the 1-norm $\|\cdot\|_1$ on $\R^n$ through
$$
\|x\|_1=|x_1|+\cdots+|x_n|,\qquad x=(x_1,\ldots,x_n)\in\R^n.
$$
We also define, for all $x=(x_1,\ldots,x_n)\in\R^n$, ${\rm supp}\,x$ as the set of those $k\in\{1,\ldots,n\}$ such that $x_k\neq0$. For $x=(x_1,\ldots,x_n)\in\R^n$ and $y=(y_1,\ldots,y_m)\in\R^m$, we define $x\oplus y\in\R^{n+m}$ and $x\otimes y\in\R^{nm}$ through
\begin{align*}
x\oplus y&=(x_1,\ldots,x_n,y_1,\ldots,y_m),\\
x\otimes y&=(x_iy_j)_{i,j}
\end{align*}
where, in the latter formula, we identify $\{1,\ldots,n\}\times\{1,\ldots,m\}$ with $\{1,\ldots,nm\}$ in an unspecified way which we keep fixed throughout this work. In the sequel, we usually restrict these definitions for $x\in\R_+^d$.

\subsection{Preordered semirings and a Vergleichsstellensatz}

Preordered semirings are essentially otherwise like commutative rings with the exception that the elements of a (preordered) semiring typically do not have additive inverses and there is a preorder which respects the algebraic structure. To formalize this, a preordered semiring is an ordered tuple $(S,0,1,+,\cdot,\rgeq)$ where $S\neq\emptyset$, $(S,0,+)$ and $(S,1,\cdot)$ are commutative semigroups where the multiplication distributes over the addition, and $\rgeq$ is a preorder (a reflexive and transitive binary relation) on $S$ such that
$$
x\rgeq y\ \Rightarrow\ \left\{\begin{array}{ll}
x+a\rgeq y+a,&\forall a\in S,\\
xa\rgeq ya,&\forall a\in S.
\end{array}\right.
$$
For a preordered semiring $S$ and $x,y\in S$, we denote $x\sim y$ whenever there are $z_1,\ldots,z_n\in S$ such that
$$
x\rgeq z_1\rleq z_2\rgeq\cdots\rleq z_n\rgeq y.
$$
We say that a preordered semiring $S$ is of {\it polynomial growth} when there is $u\in S$ such that
$$
x\rgeq y\ \Rightarrow\ \exists k\in\N:\ yu^k\rgeq x.
$$
Such an element $u$ is called a {\it power universal}. A preordered semiring $S$ is {\it zerosumfree} if
$$
x+y=0\ \Rightarrow\ x=0=y.
$$
We say that a preordered semiring $S$ is a {\it preordered semidomain} if
$$
xy=0\ \Rightarrow\ x=0\ {\rm or}\ y=0
$$
and
$$
0\rgeq x\rgeq 0\ \Rightarrow\ x=0.
$$

Given two preordered semirings $(S,0_S,1_S,+,\cdot,\rgeq_S)$ and $(T,0_T,1_T,+,\cdot,\rgeq_T)$, we say that a map $\Phi:S\to T$ is a {\it monotone homomorphism} if
\begin{itemize}
\item $\Phi(0_S)=0_T$, $\Phi(1_S)=1_T$,
\item $\Phi(x+y)=\Phi(x)+\Phi(y)$ for all $x,y\in S$,
\item $\Phi(xy)=\Phi(x)\Phi(y)$ for all $x,y\in S$, and
\item $x\rgeq_S y$ $\Rightarrow$ $\Phi(x)\rgeq_T\Phi(y)$.
\end{itemize}
We say that a monotone homomorphism $\Phi:S\to T$ is {\it degenerate} if $x\rgeq y$ implies $\Phi(x)=\Phi(y)$. Otherwise $\Phi$ is {\it nondegenerate}. We are primarily interested in monotone homomorphisms $\Phi:S\to\mb K$ where $\mb K\in\{\R_+,\R_+^{\rm op},\T\R_+,\T\R_+^{\rm op}\}$. Given a monotone homomorphism $\Phi:S\to\R_+$, we say that a map $\Delta:S\to\R$ is a {\it monotone derivation at $\Phi$} if $x\rgeq y$ implies $\Delta(x)\geq\Delta(y)$, $\Delta(x+y)=\Delta(x)+\Delta(y)$ for all $x,y\in S$, and
$$
\Delta(xy)=\Delta(x)\Phi(y)+\Phi(x)\Delta(y)\qquad{\rm (Leibniz\ rule)}
$$
for all $x,y\in S$. We are only interested in monotone derivations at degenerate homomorphisms $\Phi$ in which case we may view $\Phi$ as a monotone homomorphism of $S$ to $\R_+$ or, equivalently, $\R_+^{\rm op}$.

\begin{definition}\label{def:deg}
We say that a preordered semiring $S$ is {\it of degeneracy $d$} for some $d\in\N$ if there is a surjective homomorphism $\|\cdot\|:S\to\R_{>0}^d\cup\{(0,\ldots,0)\}$ with trivial kernel such that
$$
a\rgeq b\ \Rightarrow\ \|a\|=\|b\|,\quad \|a\|=\|b\|\ \Rightarrow\ a\sim b.
$$
In this situation, we denote the component homomorphisms of $\|\cdot\|$ by $\|\cdot\|_{(k)}$ for $k=1,\ldots,d$.
\end{definition}

Let us briefly note that, if $S$ is of degeneracy $d$, we have $x\sim y$ if and only if $\|x\|=\|y\|$. Naturally, $\|x\|=\|y\|$ $\Rightarrow$ $x\sim y$. For the converse, assume that $x\sim y$ and $z_1,\ldots,z_n\in S$ are such that $x\rgeq z_1\rleq z_2\rgeq\cdots\rleq z_n\rgeq y$. The first inequality in the chain implies $\|x\|=\|z_1\|$, the second implies $\|z_1\|=\|z_2\|$, and so on until $\|z_{n-1}\|=\|z_n\|$ and $\|z_n\|=\|y\|$. Thus,
$$
\|x\|=\|z_1\|=\|z_2\|=\cdots=\|z_{n-1}\|=\|z_n\|=\|y\|.
$$

\begin{definition}\label{def:TestSpectrum}
Let $S$ be a preordered semiring of polynomial growth and of degeneracy $d$ where we fix a power universal $u$ and the vector-valued homomorphism $\|\cdot\|$ of Definition \ref{def:TestSpectrum}. We denote the set of all nondegenerate monotone homomorphisms $\Phi:S\to\mb K$ with $\mb K\in\{\R_+,\R_+^{\rm op},\T\R_+,\T\R_+^{\rm op}\}$ by $\Sigma(S,\mb K)$ and the set of all the monotone derivations $\Delta:S\to\R$ at $\|\cdot\|_{(k)}$ ($k=1,\ldots,d$) with $\Delta(u)=1$ by $\mf D^k(S)$. For $\mb K\in\{\R_+,\R_+^{\rm op},\T\R_+,\T\R_+^{\rm op}\}$, we define
$$
\mf D(S,\mb K):=\left\{\frac{\log{\Phi(\cdot)}}{\log{\Phi(u)}}\,\middle|\,\Phi\in\Sigma(S,\mb K)\right\}
$$
as a set of maps $S\setminus\{0\}\to\R$. For $k\in\{1,\ldots,d\}$, we define
$$
\mf D_k(S):=\left\{S\setminus\{0\}\ni x\mapsto\frac{\Delta(x)}{\|x\|_{(k)}}\,\middle|\,\Delta\in\mf D^k(S)\right\}.
$$
Finally, we define the {\it test spectrum} through
$$
\hat{\mf D}(S):=\mf D(S,\R_+)\cup\mf D(S,\R_+^{\rm op})\cup\mf D(S,\T\R_+)\cup\mf D(S,\T\R_+^{\rm op})\cup\mf D_1(S)\cup\cdots\cup\mf D_d(S).
$$
\end{definition}

In \cite{FritzII}, the test spectrum was defined as the union of the sets $\Sigma(S,\mb K)$ with $\mb K\in\{\R_+,\R_+^{\rm op},\T\R_+,\T\R_+^{\rm op}\}$ and $\mf D^k(S)$ for $k=1,\ldots,d$. However, it is useful for our future discussion to define the test spectrum as the set $\hat{\mf D}(S)$ defined in Definition \ref{def:TestSpectrum}. The following central result ({\it Vergleichsstellensatz}) for this work is Theorem 8.6 in \cite{FritzII}:

\begin{theorem}\label{thm:Vergleichsstellensatz}
Let $S$ be a zerosumfree preordered semidomain of polynomial growth and of degeneracy $d$ for some $d\in\N$. We fix a power universal $u\in S$ and the vector-valued homomorphism $\|\cdot\|$ of Definition \ref{def:deg} to define the test spectrum $\hat{\mf D}(S)$. Suppose that $x,y\in S\setminus\{0\}$ are such that $\|x\|=\|y\|$. If
\begin{equation}\label{eq:FritzConditions}
\Delta(x)>\Delta(y)\qquad\forall\Delta\in\hat{\mf D}(S),
\end{equation}
then
\begin{itemize}
\item[(a)] $x^nu^k\rgeq y^nu^k$ for some $k\in\N$ when $n\in\N$ is sufficiently large and
\item[(b)] $xz\rgeq yz$ for some $z\in S\setminus\{0\}$ which can be chosen according to
$$
z=u^k\sum_{\ell=0}^n x^\ell y^{n-\ell}
$$
for some $k\in\N$ and $n\in\N$ sufficiently large.
\end{itemize}
If $x$ is a power universal, we may omit the appearance of $u^k$ in items (a) and (b) above. Conversely, if condition (a) or (b) above holds, then the inequalities in \eqref{eq:FritzConditions} hold non-strictly.
\end{theorem}

Before going on, let us note that any power universal $u\in S$ has $\|u\|=(1,\ldots,1)$ as long as $\|\cdot\|$ satisfies the conditions of Definition \ref{def:deg}. This is naturally equivalent with $u\sim 1$. To see this, assume that $x,y\in S$ are such that $x\rleq y$, so that there is $k\in\N$ such that $xu^k\rgeq y$. Since $x\rleq y$, we have $\|x\|=\|y\|$ and, since $xu^k\rgeq y$, we have $\|x\|\cdot\|u\|^k=\|y\|=\|x\|$ where we have used the first observation in the final equality. From this we immediately see that $\|u\|=(1,\ldots,1)$.

\section{Extensive monotone divergences over preordered semirings}\label{sec:AddMonDiv}

Throughout this section, we work with a preordered semidomain $S$ of polynomial growth and of degeneracy $d$ with a fixed $d\in\N$. To be explicit, let $(S,0,1,\rgeq)$ be a preordered semidomain of polynomial growth with a power universal $u$ and $\|\cdot\|:S\to\{(0,\ldots,0)\}\cup\R_{>0}^d$ be as in Definition \ref{def:TestSpectrum}. We also denote by
$$
S_N:=\{x\in S\,|\,\|x\|=(1,\ldots,1)\}=\{x\in S\,|\,x\sim 1\}
$$
the set of {\it normalized elements of $S$}. In most of the applications of this theory, $S_N$ is the physically (or, rather, information-theoretically) relevant part of the preordered semiring since normalization corresponds to, e.g.,\ trace-1 positive operators or probability measures. The power universal $u$ is within this normalized set. We simply denote the test spectrum as
$$
\hat{\mf D}:=\hat{\mf D}(S)
$$
and its constituent parts as
\begin{align*}
\mf D_{\mb K}&:=\mf D(S,\mb K),\quad\mb K\in\{\R_+,\R_+^{\rm op},\T\R_+,\T\R_+^{\rm op}\}\\
\mf D_{\rm nd}&:=\mf D(S,\R_+)\cup\mf D(S,\R_+^{\rm op})\cup\mf D(S,\T\R_+)\cup\mf D(S,\T\R_+^{\rm op}),\\
\mf D_k&:=\mf D_k(S),\quad k=1,\ldots,d,\\
\mf D_0&:=\mf D_1\cup\cdots\cup\mf D_d.
\end{align*}
We use the fixed power universal $u$ in the definition of all these sets. The set $\hat{\mf D}$ is a compact Hausdorff space in the topology of pointwise comparison \cite[Proposition 8.5]{FritzII}. This means that $\hat{\mf D}$ is compact w.r.t.\ the coarsest topology where the comparison maps $f_{x,y}:\hat{\mf D}\to\R$,
\begin{equation}\label{eq:comparisonmaps}
f_{x,y}(\Delta)=\Delta(x)-\Delta(y),
\end{equation}
with $x\sim y$ (i.e.,\ $\|x\|=\|y\|$) are continuous. Especially the evaluation maps $f_x=f_{x,1}$, $f_x:\Delta\mapsto\Delta(x)$, are continuous for $x\in S_N$.

\begin{definition}\label{def:mondiv}
We say that $D:S\setminus\{0\}\to\R$ is a {\it monotone divergence} if
\begin{itemize}
\item[(i)] $D(xy)=D(x)+D(y)$ for all $x,y\in S\setminus\{0\}$ (extensivity),
\item[(ii)] $x,y\in S\setminus\{0\}$, $x\rgeq y$ $\Rightarrow$ $D(x)\geq D(y)$ (monotonicity), and
\item[(iii)] $D(u)=1$ (normalization)
\end{itemize}
We denote the set of monotone divergences by $\mf D(S)$. We often simply denote $\mf D(S)=:\mf D$ when there is no risk of confusion.
\end{definition}

\begin{remark}
The normalization condition $D(u)=1$ for all $D\in\mf D$ and the fixed power universal can also be lifted. The price to pay is that the hope of finding extreme points for the convex set $\mf D$ is lost and the statement of Theorem \ref{thm:barycentre} will change a little.
\end{remark}

One easily sees that $\hat{\mf D}\subseteq\mf D$. We will show that actually all monotone divergences are barycentres of Borel probability measures over $\hat{\mf D}$.

\subsection{Divergences as barycentres over the test spectrum}\label{subsec:barycentres}

First, we present and prove a central lemma which utilizes the characterization of asymptotic ordering in $S$ through $\hat{\mf D}$.

\begin{lemma}\label{lemma:key}
Let $D:S\setminus\{0\}\to\R_+$ be a monotone divergence. If $x,y\in S\setminus\{0\}$, $x\sim y$, are such that $\Delta(x)\geq\Delta(y)$ for all $\Delta\in\hat{\mf D}$, then also $D(x)\geq D(y)$.
\end{lemma}

\begin{proof}
Suppose first that $x,y\in S\setminus\{0\}$, $x\sim y$, are such that $\Delta(x)>\Delta(y)$ for all $\Delta\in\hat{\mf D}$. According to Theorem \ref{thm:Vergleichsstellensatz}, this means that there are $n,k\in\N$ such that $u^k x^n\rgeq u^k y^n$. Using the monotonicity of $D$, we now have
$$
kD(u)+nD(x)=D(u^k x^n)\geq D(u^k y^n)=kD(u)+nD(y)\ \Rightarrow\ D(x)\geq D(y).
$$
Assume next that $x,y\in S\setminus\{0\}$, $x\sim y$, are such that $\Delta(x)\geq\Delta(y)$ for all $\Delta\in\hat{\mf D}$. Let us make the inequality strict by adding the power universal. Indeed, for all $m\in\N$, we have
$$
\Delta(ux^m)=\Delta(u)+m\Delta(x)=1+m\Delta(x)>m\Delta(x)\geq m\Delta(y)=\Delta(y^m)
$$
for all $\Delta\in\hat{\mf D}$. According to what we proved earlier, this means that $D(ux^m)\geq D(y^m)$, i.e.,
$$
mD(x)+D(u)\geq mD(y)\ \Rightarrow\ D(x)+\frac{1}{m}D(u)\geq D(y).
$$
Since this holds for all $m\in\N$, we have $D(x)\geq D(y)$.
\end{proof}

The above lemma is the crucial ingredient in the proof of the following result; the rest of the proof is straightforward functional analysis. The proof closely follows the idea of the proof of Theorem 2 of \cite{Mu_et_al_2021}.

\begin{theorem}\label{thm:barycentre}
Let $D\in\mf D$. Denote the Borel $\sigma$-algebra of $\hat{\mf D}$ w.r.t.\ the topology of pointwise comparison by $\mc B(\hat{\mf D})$. Suppose that, for any $z\in\R_{>0}^d\cup\{(0,\ldots,0)\}$, there is $a_z\in S$ with $\|a_z\|=z$ such that $z\mapsto\Delta(a_z)$ and $z\mapsto D(a_z)$ are measurable. There is an inner and outer regular probability measure $\mu:\mc B(\hat{\mf D})\to[0,1]$ and $a_1,\ldots,a_d\in\R$ such that, for all $x\in S\setminus\{0\}$,
\begin{equation}\label{eq:barycentre}
D(x)=\int_{\hat{\mf D}} \Delta(x)\,d\mu(\Delta)+\sum_{k=1}^d a_k\log{\|x\|_{(k)}}.
\end{equation}
\end{theorem}

\begin{proof}
Denote by $C(\hat{\mf D})$ the set of continuous functions on the compact Hausdorff space $\hat{\mf D}$ (w.r.t.\ the topology of pointwise comparison) which we equip with the sup-norm $\|\cdot\|_\infty$. For all $x,y\in S\setminus\{0\}$ with $x\sim y$, recall the comparison maps $f_{x,y}$ as defined in \eqref{eq:comparisonmaps}. By the definition of the topology on $\hat{\mf D}$, these maps are continuous, i.e.,\ $f_{x,y}\in C(\hat{\mf D})$ for all $x,y\in S$, $x\sim y$. Moreover, these functions have a couple of important properties: When $x,y\in S\setminus\{0\}$ are such that $x\rgeq y$, then
$$
f_{x,y}(\Delta)=\Delta(x)-\Delta(y)\geq0
$$
for all $\Delta\in\mf D_{\rm nd}$ and similarly for $\Delta\in\mf D_0$. Thus, $f_{x,y}\geq0$. Furthermore, when $x_1,x_2,y_1,y_2\in S$, where $x_1\sim y_1$ and $x_2\sim y_2$, are such that $f_{x_1,y_1}\geq f_{x_2,y_2}$, then $\Delta(x_1y_2)\geq\Delta(x_2y_1)$ for all $\Delta\in\hat{\mf D}$. We also find that, if $x_1\sim y_1$ and $x_2\sim y_2$, then
$$
f_{x_1x_2,y_1y_2}=f_{x_1,y_1}+f_{x_2,y_2}.
$$

Our strategy is to define a subspace $\mc V$ of $C(\hat{\mf D})$ containing $f_{x,y}$ with all $x,y\in S$ with $x\sim y$ and a positive linear functional $H:\mc V\to\R$ through $H(f_{x,y})=D(x)-D(y)$ and then extend it into a positive linear functional $I:C(\hat{\mf D})\to\R$ using a result due to Kantorovich related to the Hahn-Banach theorem. Riesz representation theorem for positive linear functionals on functions vanishing at the infinity then implies the existence of the measure of the claim. The main ingredient in this proof is Lemma \ref{lemma:key}.

Let us first show that the ansatz $F(f_{x,y})=D(x)-D(y)$ for a linear functional makes sense. Denote by $\mc F\subseteq C(\hat{\mf D})$ the set of $f_{x,y}$ with $x,y\in S$, $x\sim y$. Assume that $x,y\in S$, $x\sim y$, are such that $f_{x,y}=0$, i.e.,\ $\Delta(x)=\Delta(y)$ for all $\Delta\in\hat{\mf D}$. According to Lemma \ref{lemma:key}, this means that $D(x)=D(y)$, and we may define the map $F:\mc F\to\R$ through $F(f_{x,y})=D(x)-D(y)$ for all $x,y\in S$, $x\sim y$. As we have seen above, $f_{x_1,y_1}\geq f_{x_2,y_2}$ for $x_1\sim y_1$ and $x_2\sim y_2$ implies $\Delta(x_1y_2)\geq\Delta(x_2y_1)$ for all $\Delta\in\hat{\mf D}$ which, according to Lemma \ref{lemma:key} means that $D(x_1y_2)\geq D(x_2y_1)$; note that $x_1y_2\sim x_2y_1$. This means that $F(f_{x_1,y_1})\geq F(f_{x_2,y_2})$. Thus, $F$ is monotone. It is also additive. Indeed, for all $x_1,x_2,y_1,y_2\in S$ with $x_1\sim y_1$ and $x_2\sim y_2$,
\begin{align*}
F(f_{x_1,y_1}+f_{x_2,y_2})&=F(f_{x_1y_1,x_2y_2})=D(x_1y_1)-D(x_2y_2)\\
&=D(x_1)+D(y_1)-D(x_2)-D(y_2)=F(f_{x_1,y_1})+F(f_{x_2,y_2}).
\end{align*}

Next we extend $F$ onto the cone generated by $\mc F$. Denote by $\mc C_\Q$ the cone generated by $\mc F$, i.e.,\ the set of combinations $\sum_{i=1}^n \alpha_i f_{x_i,y_i}$ with $\alpha_1,\ldots,\alpha_n\in\Q_+$, $x_i,y_i\in S$, $x_i\sim y_i$, $i=1,\ldots,n$, and $n\in\N$. Because $\mc F$ is closed under sums, it is easy to see that
$$
\mc C_\Q=\bigcup_{n=1}^\infty \frac{1}{n}\mc F.
$$
Thus, it suffices to extend $F$ for functions of the form $(1/n)f_{x,y}$. Assume that $x_1,x_2,y_1,y_2\in S$, $x_1\sim y_1$, $x_2\sim y_2$ and $m,n\in\N$ are such that $(1/n)f_{x_1,y_1}=(1/m)f_{x_2,y_2}$. This is equivalent with $f_{x_1^m,y_1^m}=f_{x_2^n,y_2^n}$, i.e.,\ $\Delta(x_1^my_2^n)=\Delta(x_2^ny_1^m)$ for all $\Delta\in\hat{\mf D}$. Since $x_1^my_2^n\sim x_2^ny_1^m$, Lemma \ref{lemma:key} implies that $D(x_1^my_2^n)=D(x_2^ny_1^m)$, i.e.,\ $(1/n)F(f_{x_1,y_1})=(1/m)F(f_{x_2,y_2})$. Thus, we may define $G:\mc C_\Q\to\R$ through $G(n^{-1}f_{x,y})=n^{-1}\big(D(x)-D(y)\big)$. It follows easily that $G$ is still monotone and additive, i.e.,\ $G(f)\geq G(g)$ whenever $f,g\in\mc C_\Q$ are such that $f\geq g$ and $G(f+g)=G(f)+G(g)$ for all $f,g\in\mc C_\Q$. Using again the fact that $\mc F$ is closed under sums, one may also show that $G(f+\alpha g)=G(f)+\alpha G(g)$ for all $f,g\in\mc C_\Q$ and $\alpha\in\Q_+$.

Next we want to extend $G$ onto the $\|\cdot\|_\infty$-closure $\overline{\mc C_\Q}$ of $\mc C_\Q$ which includes the full positive cone $\mc C$ generated by $\mc F$, i.e.,\ the set of combinations $\sum_{i=1}^n \alpha_i f_{x_i,y_i}$ with $\alpha_1,\ldots,\alpha_n\in\R_+$, $x_i,y_i\in S$, $x_i\sim y_i$, $i=1,\ldots,n$, and $n\in\N$. To do this, we show that $G$ is uniformly continuous. For $x\in S_N$ (especially $x=u$), we define $f_x:=f_{x,1}$, $f_x(\Delta)=\Delta(x)$ for all $\Delta\in\hat{\mf D}$. Let $f,g\in\mc C_\Q$. We have, for all $\Delta\in\hat{\mf D}$ and $q\in\Q$ such that $q\geq\|f-g\|_\infty$,
\begin{align*}
g(\Delta)\leq f(\Delta)+\|f-g\|_\infty=&f(\Delta)+\|f-g\|_\infty\underbrace{\Delta(u)}_{=1}\\
=&f(\Delta)+\|f-g\|_\infty f_u(\Delta)\leq f(\Delta)+qf_u(\Delta).
\end{align*}
Due to the monotonicity of $G$,
$$
G(g)\leq G(f+qf_u)=G(f)+qG(f_u)=G(f)+qD(u)=G(f)+q.
$$
Letting $q\to\|f-g\|_\infty$, we have $G(g)\leq G(f)+\|f-g\|_\infty$. Similarly, due to symmetry, we have $G(f)\leq G(g)+\|f-g\|_\infty$, i.e.,
$$
|G(f)-G(g)|\leq\|f-g\|_\infty.
$$
Thus, $G$ extends into a uniformly continuous functional $\overline{G}:\overline{\mc C_\Q}\to\R$. Especially, $\overline{G}$ is uniformly continuous on $\mc C$.

Next we show that $\overline{G}$ is still monotone and positive-linear. Assume that $f,g\in\mc C$ are such that $f\geq g$. For $i=1,2$, there are sequences $(a_{i,k})_{k=1}^\infty$ of natural numbers and $(x_{i,k})_{k=1}^\infty$ and $(y_{i,k})_{k=1}^\infty$ in $S$ with $x_{i,k}\sim y_{i,k}$ for all $k$ such that $a_{1,k}^{-1}f_{x_{1,k},y_{1,k}}\to f$ and $a_{2,k}^{-1}f_{x_{2,k},y_{2,k}}\to g$ in the $\|\cdot\|_\infty$-norm as $k\to\infty$. Thus, for all $n\in\N$, there is $k_n\in\N$ such that
$$
\frac{1}{a_{1,k}}f_{x_{1,k},y_{1,k}}\geq f-\frac{1}{2n},\quad\frac{1}{a_{2,k}}f_{x_{2,k},y_{2,k}}\leq g+\frac{1}{2n}
$$
for all $k\geq k_n$. Recalling that $f_u$ is the constant function 1 and $f-g\geq0$, we obtain
$$
\frac{1}{a_{2,k}}f_{x_{2,k},y_{2,k}}\leq f-g+\frac{1}{a_{2,k}}f_{x_{2,k},y_{2,k}}\leq\frac{1}{a_{1,k}}f_{x_{1,k},y_{1,k}}+\frac{1}{n}=\frac{1}{a_{1,k}}f_{x_{1,k},y_{1,k}}+\frac{1}{n}f_u
$$
for all $k\geq k_n$. Using the additivity and monotonicity of $G$, we have
$$
G\left(\frac{1}{a_{1,k}}f_{x_{1,k},y_{1,k}}\right)+\frac{1}{n}D(u)\geq G\left(\frac{1}{a_{2,k}}f_{x_{2,k},y_{2,k}}\right)
$$
for all $k\geq k_n$. Letting $k\to\infty$, we have $\overline{G}(f)+n^{-1}D(u)\geq\overline{G}(g)$ and, letting $n\to\infty$, we obtain $\overline{G}(f)\geq\overline{G}(g)$. Thus, $\overline{G}$ is monotone. For positive-linearity, pick $f,g\in\mc C$ and $\alpha\geq0$. Let sequences $(a_{i,k})_{k=1}^\infty$, $(x_{i,k})_{k=1}^\infty$, and $(y_{i,k})_{k=1}^\infty$ be as above and, additionally, $(\alpha_k)_{k=1}^\infty$ be a sequence in $\Q_+$ converging to $\alpha$. It is clear that
$$
\left\|f+\alpha g-\left(\frac{1}{a_{1,k}}f_{x_{1,k},y_{1,k}}+\frac{\alpha_k}{a_{2,k}}f_{x_{2,k},y_{2,k}}\right)\right\|_\infty \to0
$$
as $k\in\infty$. Thus,
\begin{align*}
\overline{G}(f+\alpha g)=&\lim_{k\to\infty} G\left(\frac{1}{a_{1,k}}f_{x_{1,k},y_{1,k}}+\frac{\alpha_k}{a_{2,k}}f_{x_{2,k},y_{2,k}}\right)\\
=&\lim_{k\to\infty}\left\{G\left(\frac{1}{a_{1,k}}f_{x_{1,k},y_{1,k}}\right)+\alpha_k G\left(\frac{1}{a_{2,k}}f_{x_{2,k},y_{2,k}}\right)\right\}\\
=&\lim_{k\to\infty}G\left(\frac{1}{a_{1,k}}f_{x_{1,k},y_{1,k}}\right)+\left(\lim_{k\to\infty}\alpha_k\right)\lim_{k\to\infty}G\left(\frac{1}{a_{2,k}}f_{x_{2,k},y_{2,k}}\right)\\
=&\overline{G}(f)+\alpha\overline{G}(g),
\end{align*}
implying that $\overline{G}$ is positive-linear.

Let us define the vector space $\mc V:=\mc C-\mc C\subseteq C(\hat{\mf D})$. We may define the linear functional $H:\mc V\to\R$ through $H(f-g)=\overline{G}(f)-\overline{G}(g)$ for all $f,g\in\mc C$ since, as an affine functional, $\overline{G}$ extends uniquely into a linear functional defined on the linear hull. It is easy to verify that $H$ is positive, i.e.,\ $H(f)\geq 0$ whenever $f\in V$ is positive ($f(\Delta)\geq0$ for all $\Delta\in\hat{\mf D}$). Indeed, we may approximate any positive element of $\mc V$ with differences $n^{-1}f_{x_1,y_1}-m^{-1}f_{x_2,y_2}\geq0$ with $m,n\in\N$, $x_1\sim y_1$, and $x_2\sim y_2$. By now familiar calculations show that $D(x_1^my_2^n)\geq D(x_2^ny_1^m)$, i.e.,\ $H(n^{-1}f_{x_1,y_1})\geq H(m^{-1}f_{x_2,y_2})$. Let $f\in C(\hat{\mf D})$ and pick $n\in\N$ such that $n\geq\|f\|_\infty$. Thus,
$$
f(\Delta)\leq\|f\|_\infty=\|f\|_\infty f_u(\Delta)\leq nf_u(\Delta)=f_{u^n}(\Delta)
$$
for all $\Delta\in\hat{\mf D}$. Thus, all elements of $C(\hat{\mf D})$ are bounded by elements of $\mc V$. According to a result due to Kantorovich \cite{Kantorovich} (a modern version also in Theorem 8.32 of \cite{InfDimAnalysis}), this means that there is a positive linear functional $I:C(\hat{\mf D})\to\R$ that extends $H$, i.e.,\ $I|_{\mc V}=H$. Riesz-Markov theorem now tells us that there is a bounded measure $\mu:\mc B(\hat{\mf D})\to\R_+$ which is both inner and outer regular such that $I(f)=\int_{\hat{\mf D}}f\,d\mu$ for all $f\in C(\hat{\mf D})$. Especially, for all $x, y\in S$ with $x\sim y$,
\begin{equation}\label{eq:apu}
D(x)-D(y)=I(f_{x,y})=\int_{\hat{\mf D}} f_{x,y}(\Delta)\,d\mu(\Delta)
=\int_{\hat{\mf D}}\big(\Delta(x)-\Delta(y)\big)\,d\mu(\Delta).
\end{equation}
Especially, whenever $x\in S_N$,
$$
D(x)=\int_{\hat{\mf D}}\Delta(x)\,d\mu(\Delta).
$$
Since $u\in S_N$ and $\Delta(u)=1=D(u)$ for all $\Delta\in\hat{\mf D}$,
$$
\mu(\mf D)=\int_{\hat{\mf D}}d\mu(\Delta)=\int_{\hat{\mf D}}\underbrace{\Delta(u)}_{=1}\,d\mu(\Delta)=D(u)=1,
$$
meaning that $\mu$ is a probability measure.

Let us pick a function $\R_{>0}^d\cup\{(0,\ldots,0)\}\ni z\mapsto a_z\in S$ with $\|a_z\|=z$ such that $z\mapsto\Delta(a_z)$ and $z\mapsto D(a_z)$ are measurable. Equation \eqref{eq:apu} implies that
$$
x\mapsto\tilde{D}(x)=D(x)-\int_{\hat{\mf D}}\Delta(x)\,d\mu(\Delta)
$$
is constant on the cosets in $S/\!\sim$. Defining $\delta:\R_{>0}^d\cup\{(0,\ldots,0)\}\to\R$ through $\delta(z)=\tilde{D}(a_z)$, we obtain from \eqref{eq:apu}
$$
D(x)=\int_{\hat{\mf D}}\Delta(x)\,d\mu(\Delta)+\delta(\|x\|)
$$
for all $x\in S$; recall that $\|\cdot\|:S\to\R_{>0}^d\cup\{(0,\ldots,0)\}$ is assumed to be surjective, so $\delta$ is fully determined. Since $D(xy)=D(x)+D(y)$ for all $x,y\in S$, it now easily follows that $\delta(vw)=\delta(v)+\delta(w)$ for all $v,w\in\R_{>0}^d\cup\{(0,\ldots,0)\}$. From the measurability assumption it easily follows that $\delta$ is also measurable. Thus, there are $a_1,\ldots,a_d\in\R$ such that $\delta(z)=a_1\log{z_1}+\cdots+a_d\log{z_d}$ for all $z=(z_1,\ldots,z_d)\in\R_{>0}^d\cup\{(0,\ldots,0)\}$ and \eqref{eq:barycentre} follows.
\end{proof}

\begin{remark}\label{rem:barycentre}
Let $D\in\mf D$ and assume that the conditions of Theorem \ref{thm:barycentre} are satisfied, so that \eqref{eq:barycentre} holds for some probability measure $\mu$ and $a_1,\ldots,a_d\in\R$. Let us note that, when $x\in S_N$ (i.e.,\ $x\sim 1$), we simply have
\begin{equation}\label{eq:barycentresimple}
D(x)=\int_{\hat{\mf D}}\Delta(x)\,d\mu(x).
\end{equation}
If we restricted all divergences on $S_N$, we could have obtained this result using a slightly simpler proof than that presented above using just the evaluation maps $f_x$ with $x\in S_N$. Note that, in this proof, we do not need the measurability assumption made in Theorem \ref{thm:barycentre}. The formula in \eqref{eq:barycentresimple} is often enough for our purposes in what follows, but the full generality of Theorem \ref{thm:barycentre} will come in handy a couple of times.

Let us also note that, if we remove the normalization condition $D(u)=1$ for monotone divergences, the measure $\mu$ of Theorem \ref{thm:barycentre} may fail to be a probability measure. Perusing the proof of Theorem \ref{thm:barycentre}, one finds that, in the absence of the normalization condition, the measure $\mu$ is an inner and outer regular finite positive measure, so only the normalization condition $\mu(\hat{\mf D})=1$ is lost.
\end{remark}

\begin{example}\label{ex:classical}
We say that $\big(p^{(1)},\ldots,p^{(d)}\big)$ is a typical finite statistical experiment of length $d$ if $p^{(k)}$ are probability vectors such that
$$
{\rm supp}\,p^{(1)}=\cdots={\rm supp}\,p^{(d)}.
$$
`Typicality' refers to the above support condition. We say that a real function $D$ defined on any typical finite statistical experiment of length $d$ is a {\it $d$-divergence} if the following are satisfied for any typical finite statistical experiments $\big(p^{(1)},\ldots,p^{(d)}\big)$ and $\big(q^{(1)},\ldots,q^{(d)}\big)$ of length $d$ and any stochastic matrix $T$ (of suitable size):
\begin{itemize}
\item[(i)] extensivity:
$$
D\big(p^{(1)}\otimes q^{(1)},\ldots,p^{(d)}\otimes q^{(d)}\big)=D\big(p^{(1)},\ldots,p^{(d)}\big)+D\big(q^{(1)},\ldots,q^{(d)}\big),
$$
\item[(ii)] data processing inequality:
$$
D\big(p^{(1)},\ldots,p^{(d)}\big)\geq D\big(Tp^{(1)},\ldots,Tp^{(d)}\big).
$$
\end{itemize}
Note that the above conditions together imply
$$
D\big(p^{(1)},\ldots,p^{(d)}\big)\geq D(1,\ldots,1)=0,
$$
where $1$ stands above for the one-entry probability vector, i.e.,\ $d$-divergences have non-negative values on typical finite statistical experiments of length $d$. One may extend the set of typical finite statistical experiments of length $d$ into a preordered semidomain $C^d$ of polynomial growth and degeneracy $d$ where the subset of statistical experiments is exactly the normalized part $C^d_N$ \cite{Farooq_et_al_2024}. The elements of the semidomain $C^d$ are certain equivalence classes of $d$-tuples $\big(p^{(1)},\ldots,p^{(d)}\big)$ of vectors $p^{(k)}\in\R_+^n$, $n\in\N$, sharing the same support. The homomorphism $\|\cdot\|$ of Definition \ref{def:deg} is given by
$$
\big\|\big(p^{(1)},\ldots,p^{(d)}\big)\big\|=\big(\big\|p^{(1)}\big\|_1,\ldots,\big\|p^{(d)}\big\|_1\big)
$$
and $\big(u^{(1)},\ldots,u^{(d)}\big)$ represents a power universal if and only if $\big\|u^{(k)}\big\|_1=1$ for $k=1,\ldots,d$ and $u^{(k)}\neq u^{(\ell)}$ whenever $k\neq\ell$. We may extend the conditions (i) and (ii) above to the whole semidomain, so that $d$-divergences can be seen as monotone divergences on the semidomain in the sense of Definition \ref{def:mondiv} minus the normalization condition. According to Remark \ref{rem:barycentre}, this means that, for any $d$-divergence $D$, we may find a positive finite measure on the test spectrum $\hat{\mf D}(C^d)$ giving $D$ as its barycentre. The numbers $a_1,\ldots,a_d\in\R$ pay no role when $D$ is evaluated on a statistical experiment.

\begin{figure}
\begin{center}
\begin{overpic}[scale=0.45,unit=1mm]{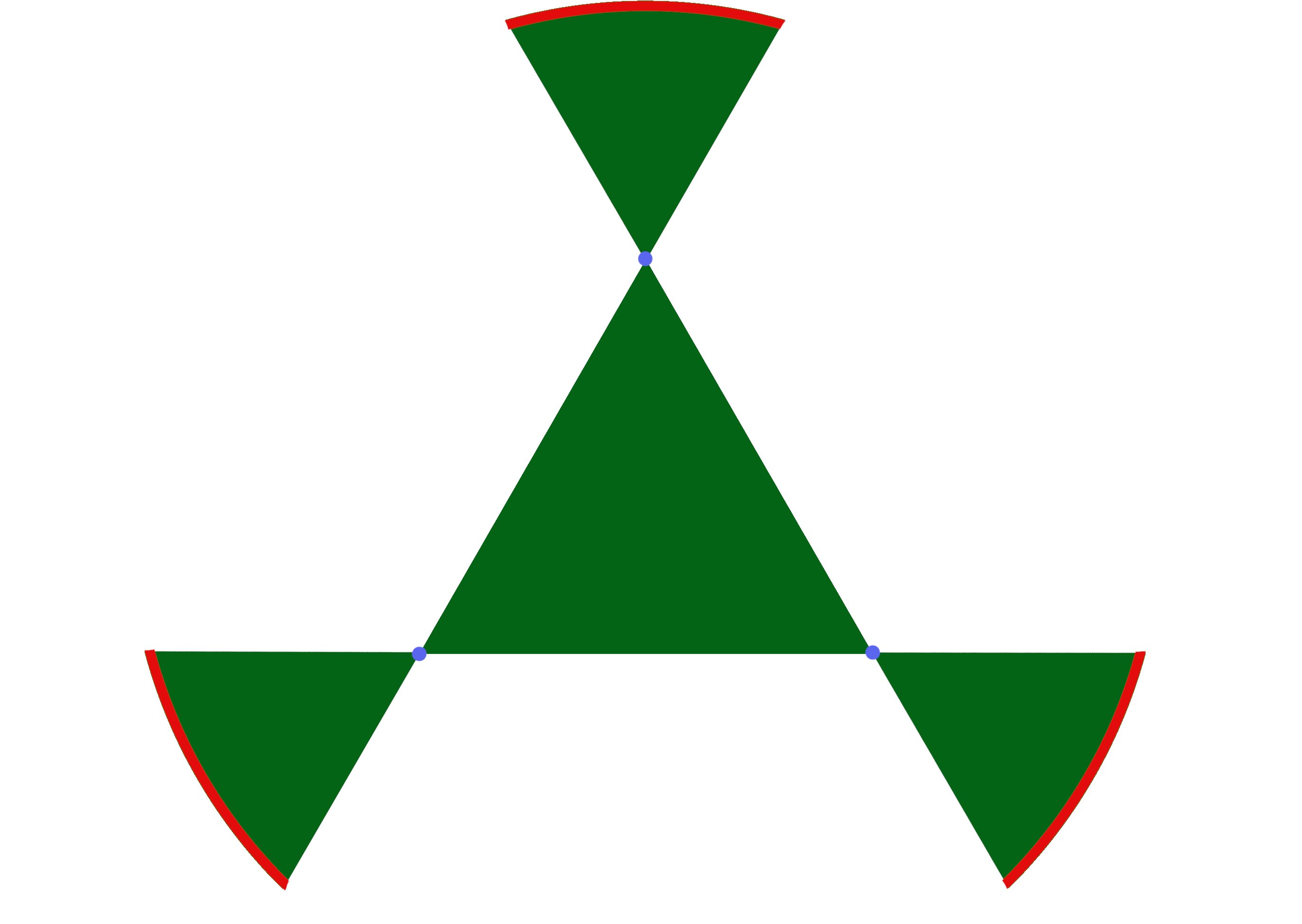}
\put(45,30){\begin{Huge}
${\color{white} A_+}$
\end{Huge}}
\put(17,12){\begin{Large}
${\color{white} A_1}$
\end{Large}}
\put(77,12){\begin{Large}
${\color{white} A_2}$
\end{Large}}
\put(47,62){\begin{Large}
${\color{white} A_3}$
\end{Large}}
\put(31,37){\begin{Large}
\darker{$D_{\underline{\alpha}}$}
\end{Large}}
\put(25,7){\begin{Large}
\darker{$D_{\underline{\alpha}}$}
\end{Large}}
\put(68,7){\begin{Large}
\darker{$D_{\underline{\alpha}}$}
\end{Large}}
\put(37,60){\begin{Large}
\darker{$D_{\underline{\alpha}}$}
\end{Large}}
\put(25,25){\begin{large}
\sininen{$\Delta^{(1)}_{\underline{\gamma}}$}
\end{large}}
\put(68,25){\begin{large}
\sininen{$\Delta^{(2)}_{\underline{\gamma}}$}
\end{large}}
\put(52,52){\begin{large}
\sininen{$\Delta^{(3)}_{\underline{\gamma}}$}
\end{large}}
\put(5,8){\begin{Large}
\punainen{$D^\T_{\underline{\beta}}$}
\end{Large}}
\put(86,8){\begin{Large}
\punainen{$D^\T_{\underline{\beta}}$}
\end{Large}}
\put(62,70){\begin{Large}
\punainen{$D^\T_{\underline{\beta}}$}
\end{Large}}
\end{overpic}
\caption{\label{fig:Restricted} The test spectrum $\hat{\mf D}(C^3)$ of the classical trivariate semiring $C^3$ depicted as a subset of the 2-dimensional affine plane of $\R^3$ of vectors whose entries sum up to 1. The part $\mf D(C^3,\R_+^{\rm op})$ consists of $D_{\underline{\alpha}}$ defined in \eqref{eq:ClTemperate} with $\underline{\alpha}\in A_+$ and $\mf D(C^3,\R_+)$ consists of $D_{\underline{\alpha}}$ with $\underline{\alpha}\in A_i$, $i=1,2,3$. Both of these temperate parts are depicted in green. The tropical part (in red) is to be understood as points at the infinity and is made up only of $\mf D(C^3,\T\R_+)$ consisting of $D^\T_{\underline{\beta}}$ as defined in \eqref{eq:ClTropical} with $\underline{\beta}\in B_i$, $i=1,2,3$; the part $\mf D(C^3,\T\R_+^{\rm op})$ is empty. The subsets $\mf D_k$, $k=1,2,3$, are depicted in blue (between the subsets $\mf D(C^3,\R_+)$ and $\mf D(C^3,\R_+^{\rm op})$ as one might expect, corresponding to the natural basis vectors of $\R^3$) and consist of $\Delta^{(k)}_{\underline{\gamma}}$ as defined in \eqref{eq:ClDeriv}. This situation generalizes to the general $d\in\N$ case in a straightforward way: the area $A_+$ corresponds to the $d$-dimensional probability simplex and $A_k$, $k=1,2,3$, extend from the extreme points of the probability simplex as cones reaching to infinity where the tropical divergences are.
}
\end{center}
\end{figure}

We may characterize the test spectrum $\hat{\mf D}(C^d)$ as follows: Let us fix a typical finite statistical experiment $\big(u^{(1)},\ldots,u^{(d)}\big)$ with $u^{(k)}\neq u^{(\ell)}$ whenever $k\neq\ell$. We define the set $A_+\subset\R^d$ as the set of those $\underline{\alpha}=(\alpha_1,\ldots,\alpha_d)$ such that $0\leq\alpha_k<1$ for $k=1,\ldots,d$ and $\alpha_1+\cdots+\alpha_d=1$. We also define the set $A_k\subset\R^d$, for $k=1,\ldots,d$, as the set of those $\underline{\alpha}=(\alpha_1,\ldots,\alpha_d)$ such that $\alpha_k>1$, $\alpha_\ell\leq0$ for $\ell\neq k$, and $\alpha_1+\cdots+\alpha_d=1$. For all $\alpha=(\alpha_1,\ldots,\alpha_d)\in A_+\cup A_1\cup\cdots\cup A_d$ we set up the function $D_{\underline{\alpha}}$ by defining it on a typical finite statistical experiment $\big(p^{(1)},\ldots,p^{(d)}\big)$ with vectors $p^{(k)}=\big(p^{(k)}_1,\ldots,p^{(k)}_n\big)\in\R_{>0}^n$ through
\begin{equation}\label{eq:ClTemperate}
D_{\underline{\alpha}}\big(p^{(1)},\ldots,p^{(d)}\big)=\frac{\log{\sum_{i=1}^n\prod_{k=1}^d\big(p^{(k)}_i\big)^{\alpha_k}}}{\log{\sum_{i=1}^n\prod_{k=1}^d\big(u^{(k)}_i\big)^{\alpha_k}}}.
\end{equation}
Furthermore, we define, for any $k=1,\ldots,d$, the set $B_k\subset\R^d$ as the set of those $\underline{\beta}=(\beta_1,\ldots,\beta_d)$ such that $\beta_k>0$, $\beta_\ell\leq0$ for all $\ell\neq k$, and $\beta_1+\cdots+\beta_d=0$. For any $\underline{\beta}=(\beta_1,\ldots,\beta_d)\in B_1\cup\cdots\cup B_d$, we set up the function $D_{\underline{\beta}}^\T$ by defining it on a typical finite statistical experiment like that above through
\begin{equation}\label{eq:ClTropical}
D_{\underline{\beta}}^\T\big(p^{(1)},\ldots,p^{(d)}\big)=\frac{\max_{1\leq i\leq n}\log{\prod_{k=1}^d\big(p^{(k)}_i\big)^{\beta_k}}}{\max_{1\leq i\leq n}\log{\prod_{k=1}^d\big(u^{(k)}_i\big)^{\beta_k}}}.
\end{equation}
These formulas can be extended to the whole semidomain $C^d$. Finally, for any $k\in\{1,\ldots,d\}$ and all $\underline{\gamma}=(\gamma_1,\ldots,\gamma_d)\in\R_+^d$, we define the function $\Delta_{\underline{\gamma}}^{(k)}$ again on the above statistical experiment through
\begin{equation}\label{eq:ClDeriv}
\Delta_{\underline{\gamma}}^{(k)}\big(p^{(1)},\ldots,p^{(d)}\big)=\frac{\sum_{\ell=1}^d\gamma_\ell D_{\rm KL}\big(p^{(k)}\big\|p^{(\ell)}\big)}{\big\|p^{(k)}\big\|_1\sum_{\ell=1}^d\gamma_\ell D_{\rm KL}\big(u^{(k)}\big\|u^{(\ell)}\big)}
\end{equation}
where $D_{\rm KL}(\cdot\|\cdot)$ is the Kullback-Leibler divergence. Note that the norm $\big\|p^{(k)}\big\|_1$ in the denominator only matters when the argument is not a statistical experiment when we extend this formula for the whole semidomain.

According to \cite{Farooq_et_al_2024}, we have
\begin{align*}
\mf D(C^d,\R_+)&=\left\{D_{\underline{\alpha}}\,\middle|\,\underline{\alpha}\in A_1\cup\cdots\cup A_d\right\},\\
\mf D(C^d,\R_+^{\rm op})&=\left\{D_{\underline{\alpha}}\,\middle|\,\underline{\alpha}\in A_+\right\},\\
\mf D(C^d,\T\R_+)&=\left\{D^\T_{\underline{\beta}}\,\middle|\,\underline{\beta}\in B_1\cup\cdots\cup B_d\right\},\\
\mf D(C^d,\T\R_+^{\rm op})&=\emptyset,\\
\mf D_k(C^d)&=\left\{\Delta^{(k)}_{\underline{\gamma}}\,\middle|\,\underline{\gamma}\in\R_+^d\setminus\{0\}\right\},\\
k&=1,\ldots,d.
\end{align*}
These parts of the spectrum are depicted in the case $d=3$ in Figure \ref{fig:Restricted}. Higher $d$ cases are similar but not so easy to visually represent. By including parts of the elements of the test spectrum (those that involve $\big(u^{(1)},\ldots,u^{(d)}\big)$) into a density function of the measure, we now have that, for any $d$-divergence $D$ there exist a finite positive measures
$$
\mu_+:\mc B(A_+)\to\R_+,\ \mu_k:\mc B(A_k)\to\R_+,\ \nu_k:\mc B(B_k)\to\R_+,\qquad k=1,\ldots,d
$$
and $\gamma_{k,\ell}\geq0$, $k,\ell=1,\ldots,d$, such that, for any typical finite statistical experiment $\big(p^{(1)},\ldots,p^{(d)}\big)$ of length $d$ with probability vectors $p^{(k)}=\big(p^{(k)}_1,\ldots,p^{(k)}_n\big)\in\R_{>0}^n$,
\begin{align*}
D\big(p^{(1)},\ldots,p^{(d)}\big)=&\int_{A_+}\log{\sum_{i=1}^n\prod_{\ell=1}^d\big(p^{(\ell)}_i\big)^{\alpha_\ell}}\,d\mu_+(\underline{\alpha})\\
&+\sum_{k=1}^d\int_{A_k}\log{\sum_{i=1}^n\prod_{\ell=1}^d\big(p^{(\ell)}_i\big)^{\alpha_\ell}}\,d\mu_k(\underline{\alpha})\\
&+\sum_{k=1}^d\int_{B_k}\max_{1\leq i\leq n}\log{\prod_{\ell=1}^d\big(p^{(\ell)}_i\big)^{\beta_\ell}}\,d\nu_k(\underline{\beta})\\
&+\sum_{k,\ell=1}^d \gamma_{k,\ell}D_{\rm KL}\big(p^{(k)}\big\|p^{(\ell)}\big).
\end{align*}

In the case where $d=2$, a 2-divergence is a dissimilarity measure between pairs of probability vectors sharing the same support which satisfies the data processing inequality and is extensive. Recall that such divergences are usually called relative entropies. Let $D_\alpha(\cdot\|\cdot)$ for $\alpha\in[1/2,\infty]$ be the R\'{e}nyi relative entropies; see \eqref{eq:RenyiRelEntr}. From the above result it follows that, for any relative entropy $D(\cdot\|\cdot)$, there exist finite measures $\mu_\pm:\mc B\big([1/2,\infty])\to\R_+$ such that \eqref{eq:RenyiBarycentre} holds for any pair $(p,q)$ of probability vectors sharing the same support. According to Theorem 2 of \cite{Mu_et_al_2021}, this result holds also for typical pairs of (sufficiently regular) probability measures $p$ and $q$.
\end{example}

\begin{remark}
It should be emphasized that the main ingredient for Theorem \ref{thm:barycentre} is Lemma \ref{lemma:key}, not Theorem \ref{thm:Vergleichsstellensatz} which we mainly use as a setting of our results and as a means to prove the pivotal Lemma \ref{lemma:key}. Let us also note that the proof of Theorem 2 of \cite{Mu_et_al_2021}, where the methods of our proof of Theorem \ref{thm:barycentre} were first used, is not real-algebraic in nature. As long as we have something like Lemma \ref{lemma:key} at our disposal, we obtain a barycentric result following the proof of Theorem \ref{thm:barycentre} or Theorem 2 of \cite{Mu_et_al_2021}.

Let us give an example of this. Denote by $\mc P$ the set of finite probability vectors of varying length. We say that $p\in\mc P$ majorizes $q\in\mc P$ if one of the following equivalent conditions is satisfied:
\begin{itemize}
\item[(a)] Reorganizing the entries of $p$ (respectively $q$) in a non-increasing order $p^\downarrow=(p_1^\downarrow,\ldots,p^\downarrow_n)$ (respectively $q^\downarrow=(q^\downarrow_1,\ldots,q^\downarrow_n)$), we have
$$
\sum_{i=1}^k p^\downarrow_i\geq\sum_{i=1}^k q^\downarrow_i
$$
for $k=1,\ldots,n$.
\item[(b)] There is a bistochasic matrix $S$ such that $Sp=q$. Recall that a matrix is bistochastic when its entries are non-negative and the entries in any of its rows or columns sum up to 1.
\end{itemize}
When either of the conditions (a) or (b) is satisfied, we denote $p\succeq q$. Let us call a map $H:\mc P\to\R_+$ a generalized entropy if it is extensive, i.e., $H(p\otimes q)=H(p)+H(q)$ for any $p,q\in\mc P$ and monotone, i.e., $p\succeq q$ implies $H(p)\leq H(q)$.

Let us introduce the generalized R\'{e}nyi entropies $H_\alpha$ for $\alpha\in[-\infty,\infty]$ (the attribute `generalized' referring to allowing the $\alpha$-parameter to be negative) by determining their values for $p\in\mc P$ with ${\rm supp}\,p=:I$:
$$
H_\alpha(p)=\left\{\begin{array}{ll}
-\min_{i\in I}\log{p_i},&\alpha=-\infty,\\
\frac{1}{1-\alpha}\log{\sum_{i\in I}p_i^\alpha},&\alpha\in(-\infty,0)\cup(0,1)\cup(1,\infty),\\
\sum_{i\in I}\log{p_i},&\alpha=0,\\
-\sum_{i\in I}p_i\log{p_i},&\alpha=1,\\
-\max_{i\in I}\log{p_i},&\alpha=\infty.
\end{array}\right.
$$
Corollary 46 of \cite{Farooq_et_al_2024} gives us a sufficient condition for the large-sample majorization of probability vectors: Let $p,q\in\mc P$. If $H_\alpha(p)<H_\alpha(q)$ for all $\alpha\in[-\infty,\infty]$, then $p^{\otimes n}\succeq q^{\otimes n}$ for any sufficiently large $n\in\N$. Proceeding just as in the proof of Lemma \ref{lemma:key}, we may show that, given $p,q\in\mc P$, we have $H(p)\leq H(q)$ for all generalized entropies $H$ if $H_\alpha(p)\leq H_\alpha(q)$ for all $\alpha\in[-\infty,\infty]$. Essentially the same proof as that for Theorem \ref{thm:barycentre} or that for Theorem 2 of \cite{Mu_et_al_2021} (Recalling also Remark \ref{rem:barycentre}) now shows that, for every generalized entropy $H$, there is an inner and outer regular finite positive measure $\mu:\mc B([-\infty,\infty])\to\R_+$ such that
$$
H(p)=\int_{-\infty}^\infty H_\alpha(p)\,d\mu(\alpha)
$$
for all $p\in\mc P$.
\end{remark}

\subsection{Extreme points of the set of divergences}\label{subsec:ext}

We immediately see that the set $\mf D$ of monotone divergences $D:S\to\R$ is convex, i.e.,\ for any $D_1,D_2\in\mf D$ and $t\in[0,1]$, the map $tD_1+(1-t)D_2$ defined through $\big(tD_1+(1-t)D_2\big)(x)=tD_1(x)+(1-t)D_2(x)$ for all $x\in S$ is also a monotone divergence. We notice that, as long as we are only interested in the restriction of $D$ onto $S_N$, $D$ is fully determined by a probability measure $\mu$ on $\hat{\mf D}$ according to \eqref{eq:barycentresimple}. We know that the extreme points of the convex set of these measures are the point measures $\delta_\Delta$ giving $\Delta\in\hat{\mf D}$ as their barycentres. However, since for one monotone divergence $D$ there might exist {\it a priori} many probability measures providing a barycentric representation for $D$, we cannot deduce that the extreme points of $\mf D$ are exactly those in $\hat{\mf D}$. Moreover, we notice that we may `disturb' any divergence with a {\it transition function}
$$
S\ni x\mapsto a_1\log{\|x\|_{(1)}}+\cdots+a_d\log{\|x\|_{(d)}}\in\R
$$
with $a_1,\ldots,a_d\in\R$, so $\mf D$ is clearly not compact in any meaningful sense. Thus the existence of extreme points is questionable. This is why, in the sequel, we implicitly view $\mf D$ modulo transition functions.

\begin{definition}
From now on we assume that there is a function $\R_{>0}^d\cup\{(0,\ldots,0)\}\ni z\mapsto a_z\in S$ such that $z\mapsto\Delta(a_z)$ is measurable for all $\Delta\in\hat{\mf D}$. We denote by $\tilde{\mf D}$ the subset of those $D\in\mf D$ modulo transition functions such that $z\mapsto D(a_z)$ is measurable.
\end{definition}

Note that, if we defined the transition functions simply as functions $\delta:\R_{>0}^d\cup\{(0,\ldots,0)\}\to\R$ such that $\delta(uv)=\delta(u)+\delta(v)$ for all $v,w\in\R_{>0}^d\cup\{(0,\ldots,0)\}$ (where the product of vectors is defined entrywise), then we could remove the measurability assumption above and consider $\tilde{\mf D}$ as the coset space of $\mf D$ over the set of these more general transition functions.

We now embark on identifying the extreme points of $\tilde{\mf D}$. First we show that, the set ${\rm ext}\,\mf D$ of extreme points of $\mf D$ is contained in $\hat{\mf D}$ (modulo transition functions). Note that, if we restrict the divergences onto $S_N$ which is usually the set of interest, we do not need the definition for the restricted set $\tilde{\mf D}$ of divergences.

\begin{lemma}\label{lemma:extlemma}
${\rm ext}\,\tilde{\mf D}\subseteq\hat{\mf D}$.
\end{lemma}

\begin{proof}
Let $D\in{\rm ext}\,\tilde{\mf D}$ and $\mu:\mc B(\hat{\mf D})\to[0,1]$ be a probability measure such that \eqref{eq:barycentre} holds where we ignore the transition part, i.e.,\ set $a_1=\cdots=a_d=0$. Let us make the counter assumption that there are $\Delta_1,\Delta_2\in{\rm supp}\,\mu$ such that $\Delta_1\neq\Delta_2$. Following the proof of Proposition 8.5 of \cite{FritzII}, we find $x_0\in S_N$ such that $\Delta_1(x_0)\neq\Delta_2(x_0)$. For $i=1,2$ and $\varepsilon>0$, we define the subsets $U_{i,\varepsilon}\subseteq\hat{\mf D}$,
$$
U_{i,\varepsilon}:=\{\Delta\in\hat{\mf D}\,|\,|\Delta_i(x_0)-\Delta(x_0)|<\varepsilon\}.
$$
Clearly, $U_{i,\varepsilon}$ are open neighbourhoods of $\Delta_i$ for all $\varepsilon>0$ and $i=1,2$. This means that $\mu(U_{i,\varepsilon})>0$ for $i=1,2$ and all $\varepsilon>0$. Because $\Delta_1(x_0)\neq\Delta_2(x_0)$, $U_{1,\varepsilon}\cap U_{2,\varepsilon}=\emptyset$ for sufficiently small $\varepsilon>0$. From now on, we only consider these sufficiently small $\varepsilon$.

Let us define, for $i=1,2$ and all $\varepsilon>0$ (sufficiently small) $D_{i,\varepsilon}:S\to\R$ through
$$
D_{i,\varepsilon}(x)=\frac{1}{\mu(U_{i,\varepsilon})}\int_{U_{i,\varepsilon}}\Delta(x)\,d\mu(\Delta)
$$
for all $x\in S$. We immediately see that $D_{i,\varepsilon}\in\mf D$. For $i=1,2$ and all $\varepsilon>0$, we may also define $D'_{i,\varepsilon}\in\mf D$ through
$$
D'_{i,\varepsilon}(x)=\frac{1}{1-\mu(U_{i,\varepsilon})}\int_{\hat{\mf D}\setminus U_{i,\varepsilon}}\Delta(x)\,d\mu(\Delta)
$$
for all $x\in S$; note that, when $\varepsilon>0$ is sufficiently small, then $U_{1,\varepsilon}$ and $U_{2,\varepsilon}$ are disjoint and they have a positive measure in $\mu$, so also $\mu(U_{i,\varepsilon})<\mu(U_{1,\varepsilon})+\mu(U_{2,\varepsilon})\leq\mu(\hat{\mf D})=1$, so the definition of $D'_{i,\varepsilon}$ makes sense. We now have
$$
D=\mu(U_{i,\varepsilon})D_{i,\varepsilon}+\big(1-\mu(U_{i,\varepsilon})\big)D'_{i,\varepsilon}
$$
for all $x\in S_N$, $\varepsilon>0$ (sufficiently small), and $i=1,2$. Since $0<\mu(U_{i,\varepsilon})<1$ for any $\varepsilon>0$ (sufficiently small), the extremality of $D$ implies that $D=D_{i,\varepsilon}$ for $i=1,2$ and any $\varepsilon>0$ (sufficiently small). Furthermore, for $i=1,2$ and $\varepsilon>0$ (sufficiently small), we have
\begin{align*}
|D_{i,\varepsilon}(x_0)-\Delta_i(x_0)|&=\left|\frac{1}{\mu(U_{i,\varepsilon})}\int_{U_{i,\varepsilon}}\Delta(x_0)\,d\mu(\Delta)-\Delta_i(x_0)\right|\\
&=\frac{1}{\mu(U_{i,\varepsilon})}\left|\int_{U_{i,\varepsilon}}\big(\Delta(x_0)-\Delta_i(x_0)\big)\,d\mu(\Delta)\right|\\
&\leq\frac{1}{\mu(U_{i,\varepsilon})}\int_{U_{i,\varepsilon}}|\Delta(x_0)-\Delta_i(x_0)|\,d\mu(\Delta)\\
&\leq\frac{1}{\mu(U_{i,\varepsilon})}\cdot\varepsilon\cdot\mu(U_{i,\varepsilon})=\varepsilon,
\end{align*}
where the second inequality follows from the definition of $U_{i,\varepsilon}$. Recall that, since $x_0\in S_N$, the expressions for $D_{i,\varepsilon}(x_0)$ simplify. Thus, for $i=1,2$, $D(x_0)=D_{i,\varepsilon}(x_0)\to\Delta_i(x_0)$ as $\varepsilon\to 0$ so that $\Delta_1(x_0)=D(x_0)=\Delta_2(x_0)\neq\Delta_1(x_0)$, a contradiction. Thus there cannot be distinct elements in the support of $\mu$, i.e.,\ ${\rm supp}\,\mu$ is a singleton, say $\{\Delta_0\}$ with some $\Delta_0\in\hat{\mf D}$. This means that $\mu$ is the point measure concentrated at $\Delta_0$, so that $D=\Delta_0\in\hat{\mf D}$.
\end{proof}

\begin{remark}
The inclusion ${\rm ext}\,\tilde{\mf D}\subseteq\hat{\mf D}$ proven in Lemma \ref{lemma:extlemma} may be strict, i.e.,\ not all elements of the test spectrum are extreme in $\tilde{\mf D}$. Note, e.g.,\ that the sets $\mf D_k(C^d)$ of derivations of the semiring discussed in Example \ref{ex:classical} essentially consist of proper convex mixtures of the (normalized) Kullback-Leibler divergences
$$
\big(p^{(1)},\ldots,p^{(d)}\big)\mapsto\frac{D_{\rm KL}\big(p^{(k)}\big\|p^{(\ell)}\big)}{\left\|p^{(k)}\right\|_1}
$$
which provide the extreme points of $\mf D_k(C^d)$. However, the non-derivation part $\mf D_{\rm nd}$ of the test spectrum is within ${\rm ext}\,\tilde{\mf D}$, as the following theorem states.
\end{remark}

\begin{theorem}\label{theor:extreme}
${\rm ext}\,\tilde{\mf D}=\mf D_{\rm nd}\cup{\rm ext}\,\mf D_1\cup\cdots\cup{\rm ext}\,\mf D_d$.
\end{theorem}

\begin{proof}
According to Lemma \ref{lemma:extlemma}, it suffices to prove that, whenever $\Delta\in\mf D_{\rm nd}$ or $\Delta\in{\rm ext}\,\mf D_k$ for $k\in\{1,\ldots,d\}$, then $\Delta$ is extreme in $\tilde{\mf D}$. We start by considering $\Delta\in\hat{\mf D}$. According to Theorem \ref{thm:barycentre}, there is a probability measure $\mu:\mc B(\hat{\mf D})\to[0,1]$ and $a_1,\ldots,a_d\in\R$ such that
\begin{equation}\label{eq:barycentreDelta}
\Delta(x)=\int_{\hat{\mf D}}\Delta'(x)\,d\mu(\Delta')+\sum_{k=1}^d a_k\log{\|x\|_{(k)}}
\end{equation}
for all $x\in S$. We keep this representation fixed and first prove the following implications:
\begin{itemize}
\item[(i)] $\Delta\in\mf D_{\R_+}$ $\Rightarrow$ ${\rm supp}\,\mu\subseteq\mf D_{\R_+}$
\item[(ii)] $\Delta\in\mf D_{\R_+^{\rm op}}$ $\Rightarrow$ ${\rm supp}\,\mu\subseteq\mf D_{\R_+^{\rm op}}$
\item[(iii)] $\Delta\in\mf D_{\T\R_+}$ $\Rightarrow$ ${\rm supp}\,\mu\subseteq\mf D_{\T\R_+}$
\item[(iv)] $\Delta\in\mf D_{\T\R_+^{\rm op}}$ $\Rightarrow$ ${\rm supp}\,\mu\subseteq\mf D_{\T\R_+^{\rm op}}$
\item[(v)] $\Delta\in\mf D_0$ $\Rightarrow$ ${\rm supp}\,\mu\subseteq\mf D_0$,
\begin{itemize}
\item[(vk)] $\Delta\in\mf D_k$ $\Rightarrow$ ${\rm supp}\,\mu\subseteq\mf D_k$, $k=1,\ldots,d$
\end{itemize}
\end{itemize}
In the proofs of these implications, we will often utilize the fact that we may view $\Delta\in\hat{\mf D}$ as an element of the test spectrum of $\ms{Frac}(S)$ when we consider entries $x/y$ with $x\sim y$; see also the proof of Proposition 8.5 \cite{FritzII}. This allows us to consider elements of the form $(u^k+1)/2$ with $k\in\Z$ and $2:=1+1\in S$; note that $\|u^k+1\|=(2,\ldots,2)=\|2\|$, so that $u^k+1\sim 2$.

Before going on with the proof, let us show that the sets $\mf D_{\mb K}$ with $\mb K\in\{\R_+,\R_+^{\rm op},\T\R_+,\T\R_+^{\rm op}\}$ and $\mf D_0$ are measurable subsets of $\hat{\mf D}$. Indeed, it is easily seen that
\begin{itemize}
\item $\Delta'\in\mf D_{\T\R_+^{\rm op}}$ $\Rightarrow$ $\Delta'\left(\frac{u+1}{2}\right)=0$,
\item $\Delta'\in\mf D_{\R_+^{\rm op}}$ $\Rightarrow$ $0<\Delta'\left(\frac{u+1}{2}\right)<\frac{1}{2}$,
\item $\Delta'\in\mf D_0$ $\Rightarrow$ $\Delta'\left(\frac{u+1}{2}\right)=\frac{1}{2}$,
\item $\Delta'\in\mf D_{\R_+}$ $\Rightarrow$ $\frac{1}{2}<\Delta'\left(\frac{u+1}{2}\right)<1$, and
\item $\Delta'\in\mf D_{\T\R_+}$ $\Rightarrow$ $\Delta'\left(\frac{u+1}{2}\right)=1$.
\end{itemize}
This means that each of the sets $\mf D_{\mb K}$ with $\mb K\in\{\R_+,\R_+^{\rm op},\T\R_+,\T\R_+^{\rm op}\}$ and $\mf D_0$ are suitable preimages in the evaluation function at $(u+1)/2$. Thus, it makes sense to restrict $\mu$ on each of these sets and
$$
1=\mu(\hat{\mf D})=\mu(\mf D_{\R_+})+\mu(\mf D_{\R_+^{\rm op}})+\mu(\mf D_{\T\R_+})+\mu(\mf D_{\T\R_+^{\rm op}})+\mu(\mf D_0).
$$
This result will come in handy also later. We discuss the separation of the different sets $\mf D_k$, $k=1,\ldots,d$, later.

Let us start by proving implication (i). Let $\Delta\in\mf D_{\R_+}$ and consider $x_k:=(u^{-k}+1)/2$ with $k\in\N$. Note that when we evaluate $\Delta$ on $x_k$, \eqref{eq:barycentreDelta} simplifies into
$$
\Delta(x_k)=\int_{\hat{\mf D}}\Delta'(x_k)\,d\mu(\Delta').
$$
Let $\Phi_\Delta\in\Sigma_{\R_+}$ be the monotone homomorphism underlying $\Delta$, i.e., $\Phi_\Delta=\exp{\big(\Delta(\cdot)\log{\Phi_\Delta(u)}\big)}$. We find that, for $\Delta'\in\mf D_{\R_+}$, $\Phi_{\Delta'}(x_k)=\big(\Phi_{\Delta'}(u)^{-k}+1\big)/2$ which converges to $1/2$ as $k\to\infty$. We obtain the same expression for $\Phi_{\Delta'}(x_k)$ when $\Delta'\in\mf D_{\R_+^{\rm op}}$ which grows limitlessly as $k\to\infty$. For $\Delta'\in\mf D_{\T\R_+}$, we find that $\Phi_{\Delta'}(x_k)=1$ and, for $\Delta\in\mf D_{\T\R_+^{\rm op}}$, $\Phi_{\Delta'}(x_k)=\Phi_{\Delta'}(u)^{-k}$. Moreover, for $\tilde{\Delta}\in\mf D_0$, we have $\tilde{\Delta}(x_k)=-k/2$. Denoting
$$
f(\Delta'):=\frac{\log{\Phi_\Delta}(u)}{\log{\Phi_{\Delta'}(u)}}
$$
for all $\Delta'\in\mf D_{\rm nd}$, we now have
\begin{align*}
\frac{1}{2}\big(\Phi_\Delta(u)^{-1}+1\big)=&\Phi_\Delta(x_k)=\exp{\big(\Delta(x_k)\log{\Phi_\Delta(u)}\big)}\\
=&\exp{\left(\log{\Phi_\Delta(u)}\int_{\hat{\mf D}}\Delta'(x_k)\,d\mu(\Delta')\right)}\\
=&\exp{\left(\int_{\mf D_{\rm nd}}f(\Delta')\log{\Phi_{\Delta'}(x_k)}\,d\mu(\Delta')\right)}\times\\
&\times\exp{\left(\log{\Phi_\Delta(u)}\int_{\mf D_0}\tilde{\Delta}(x_k)\,\d\mu(\tilde{\Delta})\right)}\\
=&\exp{\left(\int_{\mf D_{\R_+}}f(\Delta')\log{\left[\frac{1}{2}\big(\Phi_{\Delta'}(u)^{-k}+1\big)\right]}\,d\mu(\Delta')\right)}\times\\
&\times \exp{\left(\int_{\mf D_{\R_+^{\rm op}}}f(\Delta')\log{\left[\frac{1}{2}\big(\Phi_{\Delta'}(u)^{-k}+1\big)\right]}\,d\mu(\Delta')\right)}\times\\
&\times\exp{\left(\int_{\mf D_{\T\R_+}}f(\Delta')\log{1}\,d\mu(\Delta')\right)}\times\\
&\times\exp{\left(-k\int_{\mf D_{\T\R_+^{\rm op}}}f(\Delta')\log{\Phi_{\Delta'}(u)}\,d\mu(\Delta')\right)}\times\\
&\times\exp{\left(-\frac{k}{2}\int_{\mf D_0}\,d\mu(\tilde{\Delta})\right)}\\
=&\exp{\left(\int_{\mf D_{\R_+}}f(\Delta')\log{\left[\frac{1}{2}\big(\Phi_{\Delta'}(u)^{-k}+1\big)\right]}\,d\mu(\Delta')\right)}\times\\
&\times \exp{\left(\int_{\mf D_{\R_+^{\rm op}}}f(\Delta')\log{\left[\frac{1}{2}\big(\Phi_{\Delta'}(u)^{-k}+1\big)\right]}\,d\mu(\Delta')\right)}\times\\
&\times\exp{\left[-k\left(\mu(\mf D_{\T\R_+^{\rm op}})+\frac{1}{2}\mu(\mf D_0)\right)\log{\Phi_\Delta(u)}\right]}.
\end{align*}
The left-hand side of the above equation converges to $1/2$ as $k\to\infty$. For large $k\in\N$, the term in the first exponential above on the right-hand side behaves as $-(k/2)\mu(\mf D_{\R_+})\log{\Phi_\Delta(u)}$ whereas the term in the second exponential goes to $-\infty$ as $k\to\infty$ as long as $\mf D_{\R_+^{\rm op}}$ supports $\mu$, i.e.,\ $\mu(\mf D_{\R_+^{\rm op}})>0$; note that, on $\mf D_{\R_+^{\rm op}}$, $f$ is strictly negative and the logarithm term grows limitlessly as $k\to\infty$. The term in the final exponential patently goes to $-\infty$ as $k\to\infty$ if $\mu(\T\R_+^{\rm op})>0$ and $\mu(\mf D_0)>0$. Thus, we must have $\mu(\mf D_{\R_+^{\rm op}})=\mu(\mf D_{\T\R_+^{\rm op}})=\mu(\mf D_0)=0$, so that ${\rm supp}\,\mu\subseteq\mf D_{\R_+}\cup\mf D_{\T\R_+}$.

To separate $\mf D_{\T\R_+}$, consider
$$
y_k:=a_k\left(1+u+\frac{1}{2}u^2+\cdots+\frac{1}{k!}u^k\right),\quad a_k:=\left(2+\frac{1}{2}+\cdots+\frac{1}{k!}\right)^{-1}.
$$
As $k\to\infty$, we have $a_k\to 1/e$. This means that $\Phi_\Delta(y_k)\to e^{\Phi_\Delta(u)-1}$ as $k\to\infty$. All $\Delta'\in\mf D_{\R_+}$ behave in the same way. For $\Delta'\in\mf D_{\T\R_+}$, however, $\Phi_{\Delta'}(y_k)=\Phi_{\Delta'}(u)^k$ which grows exponentially in $k$. We now have
$$
\Phi_\Delta(y_k)=\exp{\left(\int_{\mf D_{\R_+}}f(\Delta')\log{\Phi_{\Delta'}(y_k)}\,d\mu(\Delta')\right)}\exp{\big(k\mu(\mf D_{\T\R_+})\log{\Phi_\Delta(u)}\big)},
$$
where the left-hand side converges to $e^{\Phi_\Delta(u)-1}$ as $k\to\infty$ while, at the same time, the last term grows exponentially if $\mu(\mf D_{\T\R_+})>0$ (the first term certainly not going to 0). Thus, we must have $\mu(\mf D_{\T\R_+})=0$, and the first implication is proven.

For the implication in (ii), the proof is essentially the same as with the implication in (i). In this case, we just consider $x_k=(u^k+1)/2$ with $k\in\N$ and continue in the same way as above.

Let us prove implication (iii). Recall the fact pointed out at the beginning of this proof that the element $(u+1)/2$ separates all the constituents $\mf D_{\mb K}$ with $\mb K\in\{\R_+,\R_+^{\rm op},\T\R_+,\T\R_+^{\rm op}\}$ and $\mf D_0$ of the spectrum. Let now $\Delta\in\mf D_{\T\R^+}$. Define measures $\mu_{\R_+}:\mc B(\mf D_{\R_+})\to[0,1]$ and $\mu_{\R_+^{\rm op}}:\mc B(\mf D_{\R_+^{\rm op}})\to[0,1]$ through
\begin{align*}
\mu_{\R_+}(A)&=\left\{\begin{array}{ll}
\frac{\mu(A)}{\mu(\mf D_{\R_+})},&\mu(\mf D_{\R_+})>0,\\
0,&\mu(\mf D_{\R_+})=0,
\end{array}\right.\quad\mu_{\R_+^{\rm op}}(B)=\left\{\begin{array}{ll}
\frac{\mu(B)}{\mu(\mf D_{\R_+^{\rm op}})},&\mu(\mf D_{\R_+^{\rm op}})>0,\\
0,&\mu(\mf D_{\R_+^{\rm op}})=0,
\end{array}\right.
\end{align*}
for all $A\in\mc B(\mf D_{\R_+})$ and $B\in\mc B(\mf D_{\R_+^{\rm op}})$. We have
\begin{align*}
1=&\Delta\left(\frac{u+1}{2}\right)=\int_{\mf D_{\T\R_+^{\rm op}}}0\,d\mu(\Delta')+\int_{\mf D_{\R_+}}\Delta'\left(\frac{u+1}{2}\right)\,d\mu(\Delta')+\int_{\mf D_0}\frac{1}{2}\,d\mu(\Delta')\\
&+\int_{\mf D_{\R_+^{\rm op}}}\Delta'\left(\frac{u+1}{2}\right)\,d\mu(\Delta')+\int_{\mf D_{\T\R_+}}1\,d\mu(\Delta')\\
=&\mu(\mf D_{\T\R_+^{\rm op}})\cdot 0+\mu(\mf D_{\R_+^{\rm op}})\underbrace{\int_{\mf D_{\R_+}}\Delta'\left(\frac{u+1}{2}\right)\,d\mu_{\R_+^{\rm op}}(\Delta')}_{<1/2}\\
&+\mu(\mf D_0)\cdot\frac{1}{2}+\mu(\mf D_{\R_+})\underbrace{\int_{\mf D_{\R_+}}\Delta'\left(\frac{u+1}{2}\right)\,d\mu_{\R_+}(\Delta')}_{<1}+\mu(\mf D_{\T\R_+})
\end{align*}
which is strictly less than 1 if $\mu(\mf D_{\T\R_+})<1$. Thus, $\mu(\mf D_{\T\R_+})=1$. Implication in (iv) is proven in essentially the same way.

Let us now assume that $\Delta\in\mf D_0$ and prove the implication in (v). Let us consider $x_k=(|k|-1+u^k)/|k|$ with $k\in\Z\setminus\{0\}$. For any derivation $\Delta'\in\mf D_0$ and $k\in\Z\setminus\{0\}$, we have $\Delta'(x_k)={\rm sign}(k)$. For $\Delta'\in\mf D_{\R_+}\cup\mf D_{\R_+^{\rm op}}$ and $k\in\Z\setminus\{0\}$, we have
$$
\Delta'(x_k)=\frac{\log{\left(\frac{1}{|k|}\big(|k|-1+\Phi_{\Delta'}(u)^k\big)\right)}}{\log{\Phi_{\Delta'}(u)}}\to\left\{\begin{array}{ll}
\infty&{\rm when}\ \Delta'\in\mf D_{\R_+}\ {\rm as}\ k\to\infty\\
&{\rm or}\ \Delta'\in\mf D_{\R_+^{\rm op}}\ {\rm as}\ k\to-\infty\\
0&{\rm when}\ \Delta'\in\mf D_{\R_+}\ {\rm as}\ k\to-\infty\\
&{\rm or}\ \Delta'\in\mf D_{\R_+^{\rm op}}\ {\rm as}\ k\to\infty.
\end{array}\right.
$$
Moreover, when $\Delta'\in\mf D_{\T\R_+}$, we have $\Delta'(x_k)=k$ when $k\geq 1$ and $\Delta'(x_k)=0$ when $k\leq -1$ and, when $\Delta'\in\mf D_{\T\R_+^{\rm op}}$, then $\Delta'(x_k)=0$ when $k\geq 1$ and $\Delta'(x_k)=k$ when $k\leq -1$. Thus, denoting by $H$ the Heaviside function (i.e.,\ $H(x)=1$ for $x\geq0$ and $H(x)=0$ otherwise),
\begin{align*}
{\rm sign}(k)=\Delta(x_k)=&\int_{\mf D_{\R_+}}\frac{\log{\left(\frac{1}{|k|}\big(|k|-1+\Phi_{\Delta'}(u)^k\big)\right)}}{\log{\Phi_{\Delta'}(u)}}\,d\mu(\Delta')\\
&+\int_{\mf D_{\R_+^{\rm op}}}\frac{\log{\left(\frac{1}{|k|}\big(|k|-1+\Phi_{\Delta'}(u)^k\big)\right)}}{\log{\Phi_{\Delta'}(u)}}\,d\mu(\Delta')\\
&+kH(k)\mu(\mf D_{\T\R_+})+kH(-k)\mu(\mf D_{\T\R_+^{\rm op}})+{\rm sign}(k)\mu(\mf D_0)
\end{align*}
where the right-hand side grows limitlessly as $k\to\infty$ if $\mu(\mf D_{\R_+})>0$ or $\mu(\mf D_{\T\R_+})>0$ and as $k\to-\infty$ if $\mu(\mf D_{\R_+^{\rm op}})>0$ or $\mu(\mf D_{\T\R_+^{\rm op}})>0$. Thus, we must have $\mu(\mf D_0)=1$, i.e.,\ ${\rm supp}\,\mu\subseteq\mf D_0$.

Let now $\Delta\in\mf D_k$ with some $k\in\{1,\ldots,d\}$. We prove the implication in (vk). The implication in (v) implies that $\mu(\mf D_0)=1$, so that
\begin{equation}\label{eq:apu2}
\Delta(x)=\sum_{\ell=1}^d\left(\int_{\mf D_\ell}\tilde{\Delta}(x)\,d\mu(\tilde{\Delta})+a_\ell\log{\|x\|_{(\ell)}}\right)
\end{equation}
for all $x\in X$. Let us denote by $\Delta'$ the original derivation giving rise to $\Delta$, i.e.,\ $\Delta'(x)=\|x\|_{(k)}\Delta(x)$ for all $x\in S$ (where it is understood that $\Delta'(0)=0$). Using the Leibniz rule, we have that
$$
\Delta'\left(\frac{x}{y}\right)=\frac{\Delta'(x)\|y\|_{(k)}-\|x\|_{(k)}\Delta'(y)}{\|y\|_{(k)}^2}.
$$
Using this, we have, for all $x\in S$,
$$
\Delta'\left(\frac{xu+1}{x+u}\right)=\frac{\|x\|_{(k)}-1}{\|x\|_{(k)}+1}.
$$
Since $\|(xu+1)/(x+u)\|_{(k)}=1$, we also have
$$
\Delta\left(\frac{xu+1}{x+u}\right)=\frac{\|x\|_{(k)}-1}{\|x\|_{(k)}+1}
$$
for all $x\in S$ for the modified derivation $\Delta\in\mf D_k$. We have similar forms for all the other derivations $\tilde{\Delta}\in\mf D_\ell$, $\ell\neq k$. Thus, from \eqref{eq:apu2}, we obtain
\begin{align*}
\frac{\|x\|_{(k)}-1}{\|x\|_{(k)}+1}&=\Delta\left(\frac{xu+1}{x+u}\right)=\sum_{\ell=1}^d\frac{\|x\|_{(\ell)}-1}{\|x\|_{(\ell)}+1}\mu(\mf D_\ell)
\end{align*}
for all $x\in S$. Let $x\in S$ be now such that $\|x\|_{(k)}=2$ and $\|x\|_{(\ell)}=1$ for all $\ell\neq k$ (recall that $\|\cdot\|:S\to\R_{>0}^d\cup\{(0,\ldots,0)\}$ is surjective). Substituting this into the above formula yields $1/3=\mu(\mf D_k)/3$, i.e.,\ $\mu(\mf D_k)=1$.

Let us now assume that $\Delta\in\mf D_{\R_+}$ and prove that ${\rm supp}\,\mu=\{\Delta\}$. We first show that $\nu:\mc B(\mf D)\to[0,1]$, $\nu(B)=\int_B f(\Delta')\,d\mu(\Delta')$, is a probability measure. We now know that $\mu$ is supported on $\mf D_{\R_+}$ where $f\geq0$, so that $\nu$ is a positive measure. Using again $x_k=(u^{-k}+1)/2$ for $k\in\N$, we have
\begin{align*}
\log{\frac{1}{2}}&=\lim_{k\to\infty}\log{\Phi_\Delta(x_k)}=\lim_{k\to\infty}\int_{\mf D_{\R_+}}f(\Delta')\log{\Phi_{\Delta'}(x_k)}\,d\mu(\Delta')\\
&=\int_{\mf D_{\R_+}}f(\Delta')\lim_{k\to\infty}\log{\Phi_{\Delta'}(x_k)}\,d\mu(\Delta')\\
&=\left(\log{\frac{1}{2}}\right)\int_{\mf D_{\R_+}}f(\Delta')\,d\mu(\Delta')=\left(\log{\frac{1}{2}}\right)\nu(\mf D_{\R_+}),
\end{align*}
where we have used the monotone convergence theorem in the third equality. Thus, $\nu(\hat{\mf D})=\nu(\mf D_{\R_+})=1$, so that $\nu$ is indeed a probability measure.

Using Jensen's inequality, we have, for all $x\in S_N$,
\begin{align}
\Phi_\Delta(x)&=\exp{\log{\Phi_\Delta(x)}}=\exp{\left(\int_{\mf D_{\R_+}}f(\Delta')\log{\Phi_{\Delta'}(x)}\,d\mu(\Delta')\right)}\nonumber\\
&=\exp{\left(\int_{\mf D_{\R_+}}\log{\Phi_{\Delta'}(x)}\,d\nu(\Delta')\right)}\nonumber\\
&\leq\int_{\mf D_{\R_+}}\exp{\log{\Phi_{\Delta'}(x)}}\,d\nu(\Delta')=\int_{\mf D_{\R_+}}\Phi_{\Delta'}(x)\,d\nu(\Delta').\label{eq:Jensen}
\end{align}
Let us now consider $x(s):=sx+(1-s)1$ for a fixed $x\in S_N$ and $s\in[0,1]\cap\Q$. Naturally, the following derivative can be equivalently viewed as the limit $\lim_{k\to\infty} k\big[\Phi_\Delta\big(x(1/k)\big)-1\big]$. With straightforward calculation, we find that the easily seen condition
$$
\left.\frac{d}{ds}\Phi_\Delta\big(x(s)\big)\right|_{s=0}=\Phi_\Delta(x)-1
$$
together with
$$
\Phi_\Delta\big(x(s)\big)=\exp{\left(\int_{\mf D_{\R_+}}\log{\Phi_{\Delta'}\big(x(s)\big)}\,d\nu(\Delta')\right)}
$$
implies (together with the Leibniz integral rule)
$$
\Phi_\Delta(x)-1=\int_{\mf D_{\R_+}}\big(\Phi_{\Delta'}(x)-1\big)\,d\nu(\Delta')=\int_{\mf D_{\R_+}}\Phi_{\Delta'}(x)\,d\nu(\Delta')-1.
$$
This means that we actually have an equality in \eqref{eq:Jensen}. This means that $\nu$ is supported by a set $B\in\mc B(\hat{\mf D})$ such that $A\ni a\mapsto\exp{a}$ is affine, where $A\subset\R$ is the image of the map $B\ni\Delta'\mapsto\Phi_{\Delta'}(x)\in\R$. Thus, $A$ must be a singleton. This means that all $\Phi_{\Delta'}(x)$, $\Delta'\in{\rm supp}\,\nu$, coincide. This holds for all $x\in S_N$, so all $\Delta'$ in the support of $\nu$ coincide. Thus, this support is a singleton, obviously ${\rm supp}\,\mu={\rm supp}\,\nu=\{\Delta\}$. Especially, $\Delta$ is not a convex combination of different elements of $\mf D_{\R_+}$ nor of the rest of $\hat{\mf D}$, i.e.,\ $\Delta\in{\rm ext}\,\tilde{\mf D}$. The case for $\Delta\in\mf D_{\R_+^{\rm op}}$ is dealt with in a similar fashion.

Let us now assume that $\Delta\in\mf D_{\T\R_+}$. For all $x,y\in S$, we have
\begin{align*}
\Delta\left(\frac{x+y}{2}\right)&=\frac{\log{\Phi_\Delta\left(\frac{x+y}{2}\right)}}{\log{\Phi_\Delta(u)}}=\frac{\log{\max\{\Phi_\Delta(x),\Phi_\Delta(y)\}}}{\log{\Phi_\Delta(u)}}\nonumber\\
&=\frac{\max\{\log{\Phi_\Delta(x)},\log{\Phi_\Delta(y)}\}}{\log{\Phi_\Delta(u)}}=\max\{\Delta(x),\Delta(y)\}.
\end{align*}
Moreover, for all $x\in S$,
$$
\Delta(x)=\int_{\mf D_{\T\R_+}}\Delta'(x)\,d\mu(\Delta')+\sum_{k=1}^d a_k\log{\|x\|_{(k)}},
$$
since ${\rm supp}\,\mu\subseteq\mf D_{\T\R_+}$. We shall subsequently use these facts.

Let us make the counter assumption that there is $\Delta_1\in\mf D_{\T\R_+}\setminus\{\Delta\}$ such that $\Delta_1\in{\rm supp}\,\mu$. Again following the proof of Proposition 8.5 of \cite{FritzII}, we find $x_0\in S_N$ (i.e.,\ $x_0\sim 1$), such that $\Delta(x_0)\neq\Delta_1(x_0)$. We may assume that $\Delta(x_0)>\Delta_1(x_0)$; the other case is treated similarly to what follows. We are free to assume that $\Delta(x_0)>0>\Delta_1(x_0)$. Indeed, if this is not already the case, we can choose $m,n\in\N$ such that $\Delta(x_0)>n/m>\Delta_1(x_0)$. Replacing $x_0$ with $x'_0:=x^m u^{-n}$, the desired inequalities hold.

We may define the sets
$$
V_\varepsilon:=\big\{\Delta'\in\mf D_{\T\R_+}\,\big|\,|\Delta_1(x_0)-\Delta'(x_0)|<\varepsilon\big\},\quad\varepsilon>0,
$$
which are clearly neighbourhoods of $\Delta_1$. As $\Delta_1\in{\rm supp}\,\mu$, we have $\mu(V_\varepsilon)>0$ for all $\varepsilon>0$ and, for suitably small $\varepsilon>0$ (namely, when $\Delta\notin V_\varepsilon$), $\mu(V_\varepsilon)<1$. Similarly as in the proof of Lemma \ref{lemma:extlemma}, we may define divergences $D_\varepsilon$ and $D'_\varepsilon$,
\begin{align*}
D_\varepsilon(x)&=\frac{1}{\mu(V_\varepsilon)}\left(\int_{V_\varepsilon}\Delta'(x)\,d\mu(\Delta')+\sum_{k=1}^d a_k\log{\|x\|_{(k)}}\right),\\
D'_\varepsilon(x)&=\frac{1}{\mu(\mf D_{\T\R_+}\setminus V_\varepsilon)}\left(\int_{\mf D_{\T\R_+}\setminus V_\varepsilon}\Delta'(x)\,d\mu(\Delta')+\sum_{k=1}^d a_k\log{\|x\|_{(k)}}\right),
\end{align*}
so that $\Delta=\mu(V_\varepsilon)D_\varepsilon+\big(1-\mu(V_\varepsilon)\big)D'_\varepsilon$. We now have that $|D_\varepsilon(x_0)-\Delta_1(x_0)|<\varepsilon$. Especially, when $\varepsilon>0$ is small enough, we have that $D_\varepsilon(x_0)<0$ since this number is $\varepsilon$-close to $\Delta_1(x_0)<0$. Thus, we find $t\in(0,1)$ and divergences $D$ and $D'$ with $tD+(1-t)D'=\Delta$ where $\Delta\neq D$ and, even more, $D(x_0)<0$. Moreover, there is a set $B\in\mc B(\hat{\mf D})$ with $t=\mu(B)$ and
\begin{align*}
D(x)&=\frac{1}{t}\left(\int_B\Delta'(x)\,d\mu(\Delta')+\sum_{k=1}^d a_k\log{\|x\|_{(k)}}\right),\\
D(x)&=\frac{1}{1-t}\left(\int_{\mf D_{\T\R_+}\setminus B}\Delta'(x)\,d\mu(\Delta')+\sum_{k=1}^d a_k\log{\|x\|_{(k)}}\right).
\end{align*}

We have, for all $x,y\in S_N$,
\begin{align*}
D\left(\frac{x+y}{2}\right)&=\frac{1}{t}\int_B \Delta'\left(\frac{x+y}{2}\right)\,d\mu(\Delta')\\
&=\frac{1}{t}\int_B \max\{\Delta'(x),\Delta'(y)\}\,d\mu(\Delta')\\
&\geq\frac{1}{t}\max\left\{\int_B \Delta'(x)\,d\mu(\Delta'),\,\int_B \Delta'(y)\,d\mu(\Delta')\right\}\\
&=\max\{D(x),D(y)\}.
\end{align*}
Similarly, $D'\big((x+y)/2\big)\geq\max\{D'(x),D'(y)\}$. Thus, with $x,y\in S_N$,
\begin{align*}
\max\{\Delta(x),\Delta(y)\}&=\Delta\left(\frac{x+y}{2}\right)=tD\left(\frac{x+y}{2}\right)+(1-t)D'\left(\frac{x+y}{2}\right)\\
&\geq t\max\{D(x),D(y)\}+(1-t)\max\{D'(x),D'(y)\}\\
&\geq \max\big\{tD(x)+(1-t)D'(x),\,tD(y)+(1-t)D'(y)\big\}\\
&=\max\{\Delta(x),\Delta(y)\},
\end{align*}
implying that we have equalities in all the inequalities above. Especially, $D\big((x_0+1)/2\big)=\max\{D(x_0),0\}=0$. Similarly, $D'\big((x_0+1)/2\big)=\max\{D'(x_0),0\}$. We now have
\begin{align*}
\Delta(x_0)&=\max\{\Delta(x_0),0\}=\Delta\left(\frac{x_0+1}{2}\right)=tD\left(\frac{x_0+1}{2}\right)+(1-t)D'\left(\frac{x_0+1}{2}\right)\\
&=(1-t)\max\{D'(x_0),0\}=(1-t)D'(x_0),
\end{align*}
where the final equality is due to the fact that, if $0\geq D'(x_0)$, the final equality would give the contradictory $0<\Delta(x_0)=0$. Combining the above with
$$
\Delta(x_0)=tD(x_0)+(1-t)D'(x_0),
$$
we arrive at the contradictory $0>D(x_0)=0$. Thus, ${\rm supp}\,\mu$ must be a singleton, i.e.,\ $\mu=\delta_\Delta$. Consequently, $\Delta\in{\rm ext}\,\hat{\mf D}$. The case for $\Delta\in\mf D_{\T\R_+^{\rm op}}$ is dealt with similarly.
\end{proof}

\section{Implications for quantum divergences and relative entropies}\label{sec:quantum}

Theorems \ref{thm:Vergleichsstellensatz}, \ref{thm:barycentre}, and \ref{theor:extreme} can be applied in many settings, including the case of the comparison of finite statistical experiments which was discussed in Example \ref{ex:classical}. These results may also be applied in a quantum-theoretic counterpart of this classical problem which we shall call {\it quantum majorization}. For the duration of this section, we fix $d\in\N$ and study tuples $\vec{\rho}=\big(\rho^{(k)}\big)_{k=1}^d$ of positive semi-definite operators $\rho^{(k)}$ on a finite-dimensional Hilbert space. Typically these operators are quantum states, i.e.,\ they are of unit trace. From now on, by `positive operator' we mean `positive semi-definite operator'.

\begin{definition}\label{def:DQ}
We say that a tuple $\vec{\rho}=\big(\rho^{(k)}\big)_{k=1}^d$ of positive operators {\it majorizes} another tuple $\vec{\sigma}=\big(\sigma^{(k)}\big)_{k=1}^d$ of positive operators if there is a quantum channel (completely positive trace-preserving linear map) $\Phi$ such that $\Phi\big(\rho^{(k)}\big)=\sigma^{(k)}$ for $k=1,\ldots,d$. The finite-dimensional Hilbert space where $\sigma^{(k)}$ reside may differ from the finite-dimensional Hilbert space where $\rho^{(k)}$ reside. In this case, we denote $\vec{\rho}\succeq\vec{\sigma}$. Denote by $\mc T_n$ the set of non-zero positive operators on $\C^n$ for all $n\in\N$. We denote by $\mf D_Q^d$ the set of maps $D:\bigcup_{n=1}^\infty\mc T_n^d\to\R$ satisfying the following conditions:
\begin{itemize}
\item[(a)] For any tuples $\vec{\rho}=\big(\rho^{(k)}\big)_{k=1}^d$ and $\vec{\sigma}=\big(\sigma^{(k)}\big)_{k=1}^d$ of positive operators, we have
$$
D\big(\rho^{(1)}\otimes\sigma^{(1)},\ldots,\rho^{(d)}\otimes\sigma^{(d)}\big)=D(\vec{\rho})+D(\vec{\sigma})\quad\textrm{(extensivity).}
$$
\item[(b)] Whenever $\vec{\rho}\succeq\vec{\sigma}$, we have $D(\vec{\rho})\geq D(\vec{\sigma})$ (data processing inequality).
\end{itemize}
We call elements of $\mf D_Q^d$ {\it $d$-variate quantum divergences}. In the $d=2$ case, we call these {\it quantum relative entropies} since this term has already been firmly established in the bivariate case.
\end{definition}

The set $\mf D_Q^2$ (the bivariate case) has been extensively studied. This set includes, e.g.,\ the following relative entropies:
\begin{enumerate}
\item {\it $\alpha$-$z$ relative entropies} $D_{\alpha,z}$ \cite{Audenaert_Datta_2015,Hiai_Jencova_2024,Jaksic2012},
\begin{equation}\label{eq:alphazed}
D_{\alpha,z}(\rho\|\sigma)=\frac{1}{\alpha-1}\log{\tr{\left(\rho^{\frac{\alpha}{2z}}\sigma^{\frac{1-\alpha}{z}}\rho^{\frac{\alpha}{2z}}\right)^z}}.
\end{equation}
These always satisfy condition (a) and, according to \cite{Carlen_et_al_2018,Hiai2013,Zhang2020}, condition (b) is satisfied if and only if
\begin{equation}\label{eq:alphazedcond}
0<\alpha<1,\ \max\{\alpha,1-\alpha\}\leq z\quad{\rm or}\quad\alpha>1,\ \max\{\alpha/2,\alpha-1\}\leq z\leq\alpha.
\end{equation}
The choice $z=1$ results in the Petz-type quantum relative entropies \cite{Petz_85,Petz_1986} and $\alpha=z$ leads to the `sandwiched' quantum R\'{e}nyi relative entropies \cite{Muller-Lennert_et_al_2013}. The pointwise limit $(\alpha,z)\to(1,1)$ of these divergences coincides with the Umegaki quantum divergence $D_{\rm U}$,
\begin{equation}\label{eq:Umegaki}
D_{\rm U}(\rho\|\sigma)=\tr{\rho(\log{\rho}-\log{\sigma})}.
\end{equation}
This is also an element of $\mf D_Q^2$.
\item {\it Matrix means and the related Kubo-Ando quantum relative entropies} $D^{\rm mm}_\alpha$ \cite{Belavkin_Staszewski_82,Kubo_Ando_79},
\begin{equation}\label{eq:KuboAndo}
D^{\rm mm}_\alpha(\rho\|\sigma)=\frac{1}{\alpha-1}\log{\tr{\sigma^{1/2}\left(\sigma^{-1/2}\rho\sigma^{-1/2}\right)^\alpha\sigma^{1/2}}}.
\end{equation}
These relative entropies always satisfy condition (a) and, when $\alpha\in[0,2]$ they also satisfy (b). The expression inside the trace coincides with the unique positive operator minimizing the $\alpha$-weighted Euclidean distances to $\rho$ and $\sigma$. At the pointwise limit $\alpha\to 1$ these divergences converge to the Belavkin-Staszewski quantum relative entropy $D_{\rm BS}$,
\begin{equation}\label{eq:BS}
D_{\rm BS}(\rho\|\sigma)=-\tr{\sigma\log{\left(\sigma^{-1/2}\rho\sigma^{-1/2}\right)}}
\end{equation}
which is also an element of $\mf D_Q^2$.
\item {\it The quantum max-relative entropy} $D_{\rm max}$,
\begin{equation}\label{eq:maxdiv}
D_{\rm max}(\rho\|\sigma)=\inf\big\{\lambda\geq0\big|\rho\leq\exp{(\lambda)}\sigma\big\}.
\end{equation}
\end{enumerate}

The matrix means of item (2) above can also generalized in the general $d\in\N$ case \cite{Bhatia_Karandikar_2012}: Define the Euclidean distance $\delta_2$ between positive operators $\rho$ and $\sigma$ through
$$
\delta_2(\rho,\sigma)=\sqrt{\sum_j\left(\log{\lambda_j(\rho^{-1}\sigma)}\right)^2}
$$
where $\lambda_j(\tau)$ are the singular values of $\tau$. Given a $d$-tuple $\vec{\rho}=\big(\rho^{(k)}\big)_{k=1}^d$ of positive operators on a finite-dimensional Hilbert space and a probability vector $\underline{\alpha}=(\alpha_1,\ldots,\alpha_d)$, the (Euclidean) matrix mean $G_{\underline{\alpha}}(\vec{\rho})$ of $\vec{\rho}$ is the positive operator $G$ minimizing the expression $\sum_{k=1}^d\alpha_k\delta_2\big(G,\rho^{(k)}\big)^2$. We may now define the $d$-variate divergence $D^{\rm mm}_{\underline{\alpha}}$ through
\begin{equation}\label{eq:MatrixMeanDiv}
D^{\rm mm}_{\underline{\alpha}}(\vec{\rho})=\frac{1}{\max_{1\leq k\leq d}\alpha_k-1}\log{\tr{G_{\underline{\alpha}}(\vec{\rho})}}.
\end{equation}
These divergences satisfy both (a) and (b) whenever $\underline{\alpha}$ is a non-extreme probability vector. The additivity property (a) is easily seen and the monotonicity property (b) follows from Theorem 4.1 and the discussion of Section 5 of \cite{Bhatia_Karandikar_2012}.

A number of multivariate quantum divergences (barycentric R\'{e}nyi divergences) were introduced in \cite{mosonyi2024geometric}. Let us briefly discuss a subset of them. Consider a finite set $A$ and fix probability vectors
\begin{align}
\underline{\beta}&=(\beta_a)_{a\in A},\\
\underline{\beta}^a&=(\beta^a_1,\ldots,\beta^a_d),\quad a\in A.
\end{align}
Using these and fixing $\alpha,z\geq0$ satisfying \eqref{eq:alphazedcond}, we define the `staged' quantity $D_{\alpha,z,\underline{\beta},(\underline{\beta}^a)_{a\in A}}$,
\begin{equation}\label{eq:staged}
D_{\alpha,z,\underline{\beta},(\underline{\beta}^a)_{a\in A}}(\vec{\rho})=D_{\alpha,z}\left(\rho^{(1)}\middle\|G_{\underline{\beta}}\big(\big[G_{\underline{\beta}^a}(\vec{\rho})\big]_{a\in A}\big)\right)
\end{equation}
for all $\vec{\rho}=\big(\rho^{(1)},\ldots,\rho^{(d)}\big)$. This type of `staging' can be also done in a number of different ways too many to list here; similar definitions have also been made in \cite{BuVra2021,Furuya_et_al_2023}. Using the monotonicity properties of the matrix mean \cite{Bhatia_Karandikar_2012}, one can show that, for any quantum channel $\Phi$ and any $d$-tuple $\vec{\rho}=\big(\rho^{(1)},\ldots,\rho^{(d)}\big)$,
$$
G_{\underline{\beta}}\Big(\Big[G_{\underline{\beta}^a}\big(\Phi(\rho^{(1)}),\ldots,\Phi\big(\rho^{(d)}\big)\big)\Big]_{a\in A}\big)\geq\Phi\Big(G_{\underline{\beta}}\big(\big[G_{\underline{\beta}^a}(\vec{\rho})\big]_{a\in A}\big)\Big)
$$
Using this, one can show that the map $D_{\alpha,z,\underline{\beta},(\underline{\beta}^a)_{a\in A}}$, is monotone under the quantum majorization preorder $\preceq$. It is also easily seen to satisfy condition (b) of Definition \ref{def:DQ}, i.e.,\ $D_{\alpha,z,\underline{\beta},(\underline{\beta}^a)_{a\in A}}$ is a $d$-variate quantum divergence.

In what follows we will see that, according to Theorems \ref{thm:barycentre} and \ref{theor:extreme} we may express every divergence $D\in\mf D_Q^d$ as a barycentre over the set of extreme points of $\mf D_Q^d$ which coincides with a subset of a test spectrum of a particular preordered semiring. The bulk of the test spectrum is extreme in $\mf D_Q^d$ and its elements are characterized by an extra additivity property. Interestingly, in the case $d=2$, all the $\alpha$-$z$-relative entropies satisfying condition (b) of Definition \ref{def:DQ} and the bivariate matrix mean relative entropies $D^{\rm mm}_\alpha$ with $\alpha\in[0,2]$ are within the set of extreme points of $\mf D_Q^2$. Moreover, the $d$-variate matrix mean divergences $D^{\rm mm}_{\underline{\alpha}}$ and the `staged' divergences $D_{\alpha,z,\underline{\beta},(\underline{\beta}^a)_{a\in A}}$ are also extreme in $\mf D_Q^d$ whenever $\underline{\alpha}$ is a non-extreme probability vector.

\subsection{The quantum majorization semiring}\label{subsec:quantumsemiring}

For any $n\in\N$, we denote by $\mc Q^d_n$ the set of $d$-tuples $\vec{\rho}=\big(\rho^{(1)},\ldots,\rho^{(d)}\big)$ of $\rho^{(k)}\in\mc T_n$ with coinciding supports, i.e.,
$$
{\rm supp}\,\rho^{(1)}=\cdots={\rm supp}\,\rho^{(d)}.
$$
Furthermore, we denote
$$
\mc Q^d:=\bigcup_{n=1}^\infty \mc Q^d_n.
$$
For $\vec{\rho}=\big(\rho^{(1)},\ldots,\rho^{(d)}\big)$ and $\vec{\sigma}=\big(\sigma^{(1)},\ldots,\sigma^{(d)}\big)$, we define
\begin{align*}
\vec{\rho}\boxplus\vec{\sigma}&=\big(\rho^{(1)}\oplus\sigma^{(1)},\ldots,\rho^{(d)}\oplus\sigma^{(d)}\big),\\
\vec{\rho}\boxtimes\vec{\sigma}&=\big(\rho^{(1)}\otimes\sigma^{(1)},\ldots,\rho^{(d)}\otimes\sigma^{(d)}\big).
\end{align*}
We also denote $\vec{\rho}\succeq\vec{\sigma}$ whenever there exists a quantum channel (a completely positive trace-preserving linear map) $\Phi$ such that $\Phi\big(\rho^{(k)}\big)=\sigma^{(k)}$ for $k=1,\ldots,d$. Note that this preorder is exactly the quantum majorization order set up in Definition \ref{def:DQ}.

Suppose now that $\vec{\rho}=\big(\rho^{(1)},\ldots,\rho^{(d)}\big)\in\mc Q^d_n$ and $\vec{\sigma}=\big(\sigma^{(1)},\ldots,\sigma^{(d)}\big)\in\mc Q^d_m$. We denote $\vec{\rho}\approx\vec{\sigma}$ if there is $\ell\geq n,m$ and isometries $U:\C^n\to\C^\ell$ and $V:\C^m\to\C^\ell$ such that $U\rho^{(k)}U^*=V\sigma^{(k)}V^*$ for $k=1,\ldots,d$. We denote the $\approx$-equivalence class defined by $\vec{\rho}$ by $[\vec{\rho}]$. We lift $\boxplus$, $\boxtimes$, and $\succeq$ to the coset space $\mc Q^d/\!\approx$ to obtain the {\it quantum majorization semiring}
$$
Q^d=(\mc Q^d/\!\approx,0,1,+,\cdot,\rgeq)
$$
where $0=[(0,\ldots,0)]$ and $1=[(1,\ldots,1)]$ where the equivalence classes are determined here by elements residing in $\mc Q^d_1$. From now on, with a slight abuse of notation, we identify elements of $Q^d$ with their representatives in $\mc Q^d$ to lighten our notation. This should cause no confusion.

We easily see that $Q^d$ is a preordered semidomain. We may also define the map $\|\cdot\|:Q^d\to\R_{>0}^d\cup\{(0,\ldots,0)\}$ through
$$
\left\|\big(\rho^{(1)},\ldots,\rho^{(d)}\big)\right\|=\left(\tr{\rho^{(1)}},\ldots,\tr{\rho^{(d)}}\right).
$$
Clearly, $\|\vec{\rho}\|=(0,\ldots,0)$ if and only if $\vec{\rho}=(0,\ldots,0)$, i.e.,\ $\|\cdot\|$ has trivial kernel. If $\vec{\rho}\succeq\vec{\sigma}$, then $\|\vec{\rho}\|=\|\vec{\sigma}\|$. This is because the operators in $\vec{\sigma}$ are obtained from the operators of $\vec{\rho}$ through a trace-preserving map. Let now $\|\vec{\rho}\|=\|\vec{\sigma}\|=(a_1,\ldots,a_d)$. Since the trace is a (completely) positive linear map and $\tr{\rho^{(k)}}=a_k=\tr{\sigma^{(k)}}$ for $k=1,\ldots,d$, we have
$$
\vec{\rho}\succeq(a_1,\ldots,a_d)\preceq\vec{\sigma},
$$
i.e.,\ $\vec{\rho}\sim\vec{\sigma}$. Thus, $\|\vec{\rho}\|=\|\vec{\sigma}\|$ $\Rightarrow$ $\vec{\rho}\sim\vec{\sigma}$, and $\|\cdot\|$ satisfies the conditions of Definition \ref{def:deg}, i.e.,\ $Q^d$ is of degeracy $d$. We also know that, if $\vec{\tau}\in\mc Q^d$ is a power universal, then $\|\vec{\tau}\|=(1,\ldots,1)$. The next lemma shows that power universals indeed exist (i.e.,\ $Q^d$ is of polynomial growth) and that the power universals are plentiful and can be easily characterized.

\begin{lemma}\label{lemma:Quniversal}
The preordered semidomain $Q^d$ is of polynomial growth and $\vec{\tau}=\big(\tau^{(1)},\ldots,\tau^{(d)}\big)$ is a power universal if and only if $\tr{\tau^{(k)}}=1$ for all $k$ and $\tau^{(k)}\neq\tau^{(\ell)}$ whenever $k\neq\ell$.
\end{lemma}

\begin{proof}
Let us introduce some notation that we use throughout this proof. A map $M:\{1,\ldots,m\}\to\mc E(\hil)$ is a (finite-outcome) POVM in a Hilbert space $\hil$ ($\mc E(\hil)$ denoting the set of effect operators, i.e.,\ the operator interval $[0,\id_\hil]$) if $M(1)+\cdots+M(m)=\id_\hil$. For any POVM $M:\{1,\ldots,m\}\to\mc E(\hil)$ and any positive (trace-class) operator $\tau$, we define the vector $p^M_\tau\in\R_+^m$ according to
$$
p^M_\tau=\left(\tr{\tau M(1)},\ldots,\tr{\tau M(m)}\right).
$$
When $\tau$ is a state, $p^M_\tau$ is a probability vector describing the outcome statistics of the quantum observable described by $M$ in the state $\tau$.

Let us pick $\vec{\tau}=\big(\tau^{(1)},\ldots,\tau^{(d)}\big)\in\mc Q^d$ where $\tau^{(k)}$ are trace-one operators on a finite-dimensional Hilbert space $\hil$ and, whenever $k\neq\ell$, $\tau^{(k)}\neq\tau^{(\ell)}$. According to \cite{Li_2016} (see also \cite{Nussbaum_Szkola_2009} for the case $d=2$), there exists a sequence $\{M_n\}_{n=1}^\infty$ of POVMs $M_n:\{1,\ldots,d\}\to\mc E(\hil^{\otimes n})$ such that, for all $k\in\{1,\ldots,d\}$ and $n\in\N$,
\begin{align*}
p^{M_n}_{(\tau^{(k)})^{\otimes n}}&=\left(\tr{\big(\tau^{(k)}\big)^{\otimes n}M_n(\ell)}\right)_{\ell=1}^d\\
&=\left(\varepsilon_{k,1}^{(n)},\ldots,\varepsilon_{k,k-1}^{(n)},1-\varepsilon_{k,k}^{(n)},\varepsilon_{k,k+1}^{(n)},\ldots,\varepsilon_{k,d}^{(n)}\right)
\end{align*}
where $\varepsilon_{k,\ell}^{(n)}\leq\exp{\big(-nC(\vec{\tau})\big)}$ where, in turn,
\begin{align*}
C(\vec{\tau})&:=\min_{k,\ell:k\neq\ell}\sup_{0\leq s\leq 1}\left\{-\log{\tr{\big(\tau^{(k)}\big)^s\big(\tau^{(\ell)}\big)^{1-s}}}\right\}\\
&\geq\min_{k,\ell:k\neq\ell}\left\{-\log{\tr{\big(\tau^{(k)}\big)^{1/2}\big(\tau^{(\ell)}\big)^{1/2}}}\right\}>-\max_{k,\ell:k\neq\ell}\log{\sqrt{\tr{\tau^{(k)}}\tr{\tau^{(\ell)}}}}=0
\end{align*}
where the final strict inequality is due to the Cauchy-Schwarz inequality, strictness following from the fact that the states $\tau^{(k)}$ all differ from each other. When we view
$$
U^{(n)}:=\left(p^{M_n}_{(\tau^{(1)})^{\otimes n}},\ldots,p^{M_n}_{(\tau^{(d)})^{\otimes n}}\right)
$$
as a $(d\times d)$-matrix with columns $p^{M_n}_{(\tau^{(k)})^{\otimes n}}$, we see that $U^{(n)}$ converges to the $(d\times d)$ identity matrix $I_d$ as $n\to\infty$. Thus, when $n$ is large enough, $U^{(n)}$ enters a neighbourhood of the identity matrix where all matrices are invertible, Thus, for $n$ large enough, $U^{(n)}$ has an inverse $T^{(n)}$ whose columns we denote by $t^{(n,k)}$ for $k=1,\ldots,d$. Since $U^{(n)}$ is a stochastic matrix, the columns of $T^{(n)}$ sum up to 1. Let us quickly show this. First, defining the horizontal row $e:=(1\,\cdots\,1)$,
$$
\left(\sum_{\ell=1}^d t^{(n,k)}_\ell\right)_{k=1}^d=eT^{(n)}=eU^{(n)}T^{(n)}=e,
$$
where the second equality is due to the fact that, since the columns of $U^{(n)}$ sum up to 1, then $eU^{(n)}=e$ and the second equality is due to the fact that $U^{(n)}T^{(n)}=I_d$.

Let us consider any $\vec{\sigma}=\big(\sigma^{(1)},\ldots,\sigma^{(d)}\big)\in\mc Q^d$ where $\sigma^{(k)}$ are states on a finite-dimensional Hilbert space $\mc K$ the set of operators on which we denote by $\mc T(\mc K)$. We next show that there is $n\in\N$ such that $\vec{\tau}^{\boxtimes n}\succeq\vec{\sigma}$. Let $n$ be large enough so that $T^{(n)}$ exists. We define the linear preparation map $\Phi^{(n)}:\R_+^d\to\mc T(\mc K)$ through $\Phi^{(n)}(e_k)=\sigma_k$ for $k=1,\ldots,d$ where $e_k$ are the natural basis vectors having 1 in the $k$th slot and zeros everywhere else. We also define the preparation map $\Psi^{(n)}:\R^d\to\mc T(\mc K)$ through $\Psi^{(n)}(p)=\Phi^{(n)}(T^{(n)}p)$ for all $p\in\R^d$. $\Phi^{(n)}$ is naturally positive and $\tr{\Phi^{(n)}(p)}=p_1+\cdots+p_d$ for all $p=(p_1,\ldots,p_d)\in\R^d$. Due to the fact that the columns of $T^{(n)}$ all sum up to 1, we see that the latter condition holds also for $\Psi^{(n)}$. Since $T^{(n)}\to I_d$ as $n\to\infty$ (due to the fact that $U^{(n)}\to I_d$ as $n\to\infty$ and the continuity of the matrix inverse) we see that the absolute value of the possible negative matrix entries in $T^{(n)}=\big(t^{(n,k)}_\ell\big)_{k,\ell=1}^d$ converge to 0. We now have, for all $k\in\{1,\ldots,d\}$,
\begin{align*}
\Psi^{(n)}(e_k)&=\Phi^{(n)}(T^{(n)}e_k)=\Phi^{(n)}(t^{(n,k)})=\sum_{ell=1}^d t^{(n,k)}_\ell\Phi^{(n)}(e_\ell)=\sum_{\ell=1}^d t^{(n,k)}_\ell\sigma^{(\ell)}\geq0
\end{align*}
where the final operator inequality holds for $n$ large enough so that the negative numbers among $t^{(n,k)}_\ell$ have small enough absolute values so that the sum operator is positive semi-definite; recall that $\sigma^{(\ell)}$ all share the same support. Thus, with $n$ large enough, $\Psi^{(n)}$ is a positive state preparator. Thus, when $n$ is sufficiently large, we may define this positive state preparator $\Psi^{(n)}$ which combined to the quantum-to-classical map $\Phi_{M_n}:\rho\mapsto p^{M_n}_{\rho}$ gives us
\begin{align*}
(\Psi^{(n)}\circ\Phi_{M_n})\big((\tau^{(k)})^{\otimes n}\big)&=\Psi^{(n)}\left(p^{M_n}_{(\tau^{(k)})^{\otimes n}}\right)\\
&=\Phi^{(n)}\left(T^{(n)}p^{M_n}_{(\tau^{(k)})^{\otimes n}}\right)=\Phi^{(n)}(e_k)=\sigma^{(k)}
\end{align*}
for $k=1,\ldots,d$. Since positive linear quantum-to-classical and classical-to-quantum maps are completely positive, so is $\Lambda^{(n)}:=\Psi^{(n)}\circ\Phi_{M_n}$, i.e.,\ $\Lambda^{(n)}$ is a quantum channel (completely positive trace-preserving linear map). Thus, $\Lambda^{(n)}\big((\tau^{(k)})^{\otimes n}\big)=\sigma^{(k)}$ for $k=1,\ldots,d$, i.e.,\ $\vec{\tau}^{\boxtimes n}\succeq\vec{\sigma}$ for sufficiently large $n\in\N$.

Let us assume that $\vec{\rho}=\big(\rho^{(1)},\ldots,\rho^{(d)}\big)\in\mc Q^d$ and $\vec{\sigma}=\big(\sigma^{(1)},\ldots,\sigma^{(d)}\big)\in\mc Q^d$ are such that $\vec{\rho},\vec{\sigma}\notin 0$ and $\vec{\sigma}\succeq\vec{\rho}$. Thus, $\|\vec{\rho}\|=\|\vec{\sigma}\|=:(a_1,\ldots,a_d)\in\R_{>0}^d$. We also have
$$
\left(\frac{1}{a_1},\ldots,\frac{1}{a_d}\right)\boxtimes\vec{\sigma}\succeq\left(\frac{1}{a_1},\ldots,\frac{1}{a_d}\right)\boxtimes\vec{\rho}
$$
where the operator tuples consist of trace-one operators. Thus, we may freely assume that $\rho^{(k)}$ and $\sigma^{(k)}$ are of unit trace. Now $\vec{\rho}\succeq(1,\ldots,1)$, which combined with our earlier observation yields
$$
\vec{\tau}^{\boxtimes n}\boxtimes\vec{\rho}\succeq\vec{\tau}^{\boxtimes n}\boxtimes(1,\ldots,1)=\vec{\tau}^{\boxtimes n}\succeq\vec{\sigma}
$$
for sufficiently large $n\in\N$ showing that $\vec{\tau}$ is a power universal.

Let us show the `only if' part of the claim. Suppose that $\vec{\tau}=\big(\tau^{(1)},\ldots,\tau^{(d)}\big)$ is a power universal. We have $\|\vec{\tau}\|=(1,\ldots,1)$, meaning $\tr{\sigma^{(k)}}=1$ for $k=1,\ldots,d$. Let us make the counter assumption that there are $k,\ell\in\{1,\ldots,d\}$, $k\neq\ell$, such that $\tau^{(k)}=\tau^{(\ell)}$. Let $\vec{\rho}=\big(\rho^{(1)},\ldots,\rho^{(d)}\big)\in\mc Q^d$ be such that $\|\vec{\rho}\|=(1,\ldots,1)$ and $\rho^{(k)}\neq\rho^{(\ell)}$. Since $\vec{\rho}\succeq(1,\ldots,1)$, there is $n\in\N$ such that $\vec{\tau}^{\boxtimes n}\succeq\vec{\rho}$. This means that there is a channel $\Phi$ such that
$$
\rho^{(k)}=\Phi\big((\tau^{(k)})^{\otimes n}\big)=\Phi\big((\tau^{(\ell)})^{\otimes n}\big)=\rho^{(\ell)}\neq\rho^{(k)},
$$
a contradiction. Thus, $\vec{\tau}$ does not contain repeating operators.
\end{proof}

We now find that the preordered semiring $Q^d$ satisfies all the conditions of Theorem \ref{thm:Vergleichsstellensatz} and, thus, also Theorems \ref{thm:barycentre} and \ref{theor:extreme} can be applied to $Q^d$. The set $\mf D(Q^d)$ of divergences as defined in Definition \ref{def:mondiv} essentially coincides with the set $\mf D_Q^d$ of $d$-variate quantum divergences introduced in Definition \ref{def:DQ}. Thus, according to Theorem \ref{thm:barycentre}, all $d$-variate quantum divergences are barycentres over the test spectrum $\hat{\mf D}(Q^d)$. Moreover, according to Theorem \ref{theor:extreme}, the extreme points of $\mf D_Q^d$ coincide with the union of the non-derivation part of $\hat{\mf D}(Q^d)$ and the set of the extreme points of the sets $\mf D_k(Q^d)$, $k=1,\ldots,d$, of derivations.

We easily obtain the following result (which has already been proven in the classical case \cite{Farooq_et_al_2024,Verhagen_et_al_2024}) which states that we do not have to bother about the $\T\R_+^{\rm op}$-part of the test spectrum.

\begin{lemma}
$\mf D(Q^d,\T\R_+^{\rm op})=\emptyset$.
\end{lemma}

\begin{proof}
When $\Phi\in\Sigma(Q^d,\mb K)$, $\mb K\in\{\R_+,\R_+^{\rm op},\T\R_+,\T\R_+^{\rm op}\}$, we can prove in exactly the same way as in the proof of Proposition 13 of \cite{Farooq_et_al_2024} that, 
whenever $(a_1,\ldots,a_d)\in\R_{>0}^d$,
$$
\Phi(a_1,\ldots,a_d)=\left\{\begin{array}{ll}
a^{\alpha_1}\cdots a_d^{\alpha_d},&\underline{\alpha}\in A_1\cup\cdots\cup A_d\ {\rm and}\ \mb K=\R_+\\
&{\rm or}\ \underline{\alpha}\in A_+\ {\rm and}\ \mb K=\R_+^{\rm op},\\
a^{\beta_1}\cdots a^{\beta_d},&\underline{\beta}\in B_1\cup\cdots\cup B_d\ {\rm and}\ \mb K\in\{\T\R_+,\T\R_+^{\rm op}\}.
\end{array}\right.
$$
Especially, when $\mb K\in\{\T\R_+,\T\R_+^{\rm op}\}$, $\Phi(\lambda,\ldots,\lambda)=1$ for all $\lambda>0$.

Let now $\Phi\in\Sigma(Q^d,\T\R_+^{\rm op})$. We have, for all $n\in\N$, $\vec{\rho},\vec{\sigma}\in\mc Q^d_n$, and $t\in(0,1)$,
\begin{align}
\Phi\big(t\vec{\rho}+(1-t)\vec{\sigma}\big)&\geq\Phi\big(t\vec{\rho}\boxplus(1-t)\vec{\sigma}\big)=\max\big\{\Phi(t\vec{\rho}),\Phi\big((1-t)\vec{\sigma}\big)\big\}\nonumber\\
&=\max\big\{\Phi(t,\ldots,t)\Phi(\vec{\rho}),\Phi(1-t,\ldots,1-t)\Phi(\vec{\sigma})\big\}\nonumber\\
&=\max\big\{\Phi(\vec{\rho}),\Phi(\vec{\sigma})\big\}.\label{eq:quasiconc}
\end{align}
We next show that $\Phi(\vec{\rho})$ only depends on the common support of the operators in $\vec{\rho}$. Let $\vec{\rho}=\big(\rho^{(1)},\ldots,\rho^{(d)}\big)$ and $\vec{\sigma}=\big(\sigma^{(1)},\ldots,\sigma^{(d)}\big)$ be elements of $\mc Q^d_n$ for some $n\in\N$ such that
$$
{\rm supp}\,\rho^{(1)}=\cdots={\rm supp}\,\rho^{(d)}={\rm supp}\,\sigma^{(1)}=\cdots={\rm supp}\,\sigma^{(d)}=:P
$$
where we view $P$ as a projection which we assume to be non-zero. There are $\vec{\tau_s}=\big(\tau_s^{(1)},\ldots,\tau_s^{(d)}\big)\in\mc Q^d_n$ with $s=1,2$ and $s,t\in(0,1)$ such that $\vec{\rho}=s\vec{\sigma}+(1-s)\vec{\tau_1}$ and $\vec{\sigma}=t\vec{\rho}+(1-t)\vec{\tau_2}$. Using \eqref{eq:quasiconc}, we may evaluate
\begin{align*}
\Phi(\vec{\sigma})&\geq\max\{\Phi(\vec{\rho}),\Phi(\vec{\tau_1})\}\geq\Phi(\vec{\rho})\geq\max\{\Phi(\vec{\sigma}),\Phi(\vec{\tau_2})\}\geq\Phi(\vec{\sigma}).
\end{align*}
Thus, $\Phi(\vec{\rho})=\Phi(\vec{\sigma})$. Denoting $\tr{P}=:m$, we now have
\begin{align*}
\Phi(\vec{\rho})&=\Phi\left(\frac{1}{m}P,\ldots,\frac{1}{m}P\right)=\Phi\left(\frac{1}{m},\ldots,\frac{1}{m}\right)=1
\end{align*}
Thus, $\Phi(\vec{\rho})=1$ whenever $\vec{\rho}$ is non-zero. Naturally, $\Phi(0)=0$. We have now seen that this is the only monotone homomorphism of $Q^d$ into $\T\R_+^{\rm op}$. This homomorphism is clearly degenerate, so that $\Sigma(Q^d,\T\R_+^{\rm op})=\emptyset=\mf D(Q^d,\T\R_+^{\rm op})$.
\end{proof}

We now know that the quantum semiring $Q^d$ satisfies the requirements of Theorem \ref{thm:Vergleichsstellensatz} and, thus, also the barycentric result Theorem \ref{thm:barycentre} and the extremality result Theorem \ref{theor:extreme} hold for the quantum divergences. We also know that the part $\hat{\mf D}(Q^d,\T\R_+^{\rm op})$ of the test spectrum is empty. However, the rest of the test spectrum remains elusive.

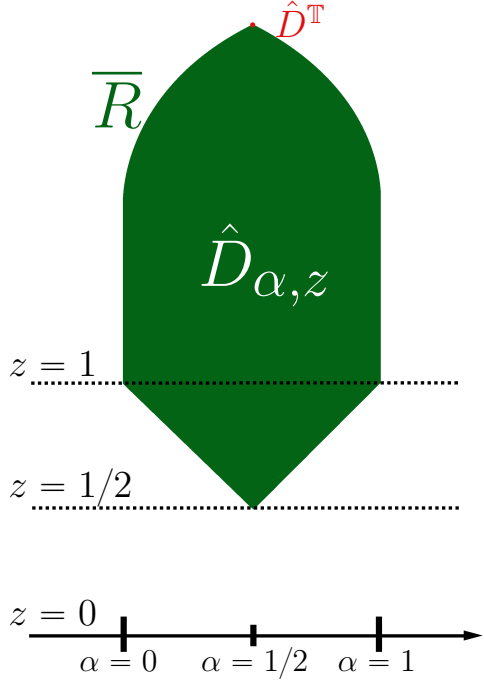
\begin{SCfigure}
\begin{overpic}[scale=0.4,unit=1mm]{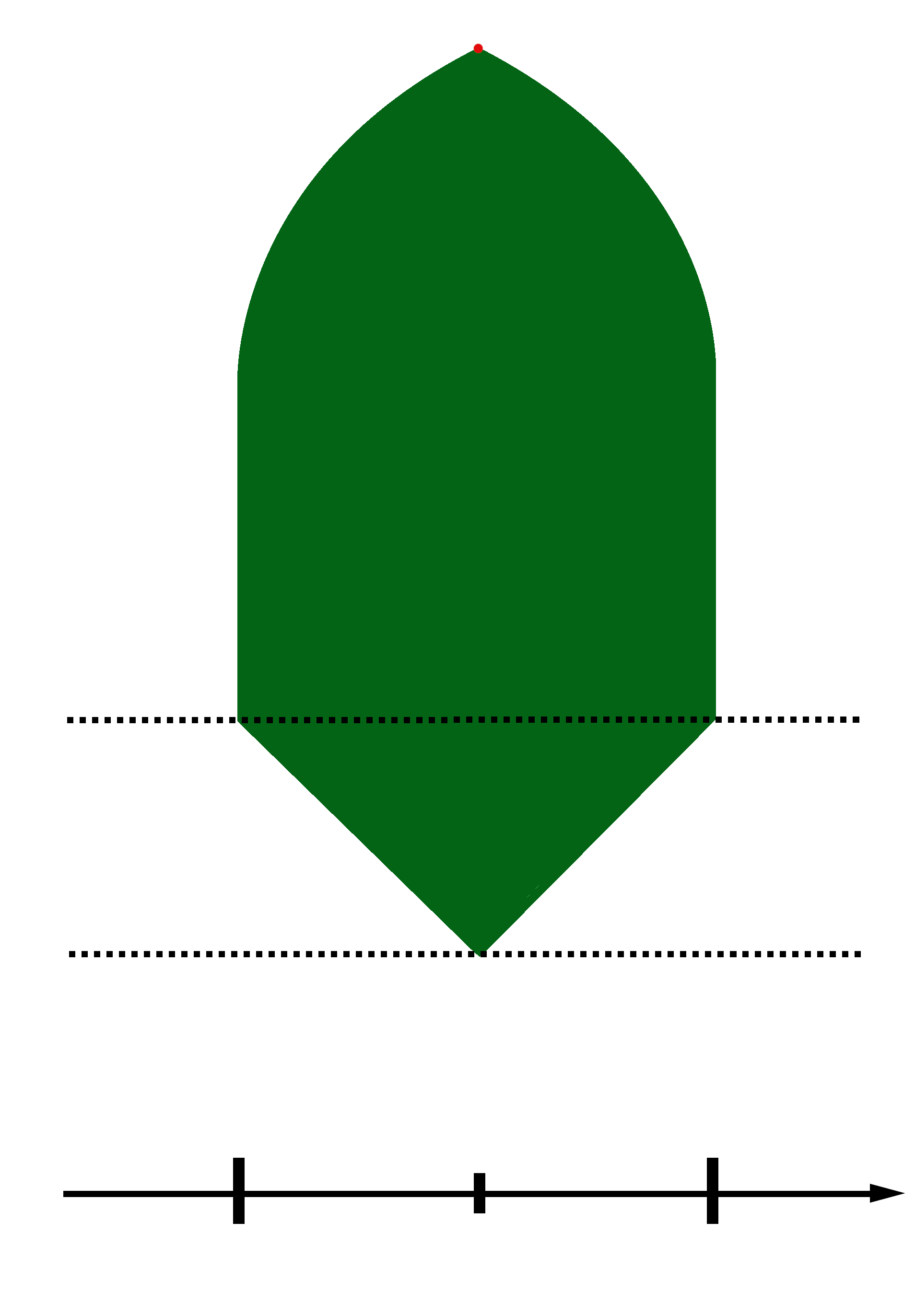}
\put(0,43){\begin{large}
$z=1$
\end{large}}
\put(0,27){\begin{large}
$z=1/2$
\end{large}}
\put(0,10){\begin{large}
$z=0$
\end{large}}
\put(11,4){$\alpha=0$}
\put(27,4){$\alpha=1/2$}
\put(45,4){$\alpha=1$}
\put(24,55){\begin{Huge}
$\textcolor{white}{\hat{D}_{\alpha,z}}$
\end{Huge}}
\put(35,88){\begin{large}
\punainen{$\hat{D}^\T$}
\end{large}}
\put(10,76){\begin{Huge}
\darker{$\overline{R}$}
\end{Huge}}
\end{overpic}
\caption{\label{fig:alphazed} A depiction of the test spectrum of the quantum semiring of pairs of states having the form of \eqref{eq:specialform}. The one-point compactification $\overline{R}$ of the set $R$ of those $(\alpha,z)$ such that $0\leq\alpha\leq 1$ and $z\geq\max\{\alpha,1-\alpha\}$ is presented in green. The bulk of this set consists of $\R_+^{\rm op}$- quantum relative entropies $\hat{D}_{\alpha,z}$ while the single point at $z=\infty$ is the $\T\R_+^{\rm op}$-quantum relative entropy $\hat{D}^\T$ (depicted in red) which is the pointwise limit of $\hat{D}_{\alpha,z}$ as $z\to\infty$.}
\end{SCfigure}

\begin{example}
Let us briefly mention a case where we can fully characterize the relevant test spectrum in a restricted quantum case.
In this bivariate setting the states simultaneously block-diagonalize the substates in the blocks being pure \cite{Verhagen_et_al_2025}. This means that the pairs of states we study have the form
\begin{equation}\label{eq:specialform}
\rho=\bigoplus_{i=1}^n p_i|\fii_i\>\<\fii_i|,\quad\sigma=\bigoplus_{i=1}^n q_i|\psi_i\>\<\psi_i|,
\end{equation}
where $(p_i)_i$ and $(q_i)_i$ are probability vectors sharing at least some common support and $\fii_i,\psi_i\in\hil_i$ are unit vectors in some Hilbert space $\hil_i$. In this special case, the test spectrum consists of the renormalized $\alpha$-$z$ quantum relative entropies $\hat{D}_{\alpha,z}:=\frac{1-\alpha}{z+1}D_{\alpha,z}$ for $0\leq\alpha<1$ and $z\geq\max\{\alpha,1-\alpha\}$ and their pointwise limits $\alpha\to 1$ (with any fixed $z$) and $z\to\infty$ (for any $\alpha$).

However, the preordered semiring involved in this setting dramatically differs from the semiring $Q^d$ studied earlier in this section: The states need not share the same support and the parts of the spectrum involving $\R_+$ and $\T\R_+$ as well as the derivation part are empty. However, there is an opposite-tropical element $\hat{D}^\T$ of the test spectrum which is given by the pointwise limit $z\to\infty$; recall that this part $\mf D(Q^d,\T\R_+^{\rm op})$ is empty for the semiring $Q^d$.

Since, according to \cite{Verhagen_et_al_2025}, the {\it Vergleichsstellensatz}, Theorem \ref{thm:Vergleichsstellensatz}, applies to this restricted quantum semiring, theorems \ref{thm:barycentre} and \ref{theor:extreme} are also applicable to this semiring. Denoting the the set of those $(\alpha,z)\in\R^2$ such that $0\leq\alpha\leq 1$ and $z\geq\max\{\alpha,1-\alpha\}$ by $R$, we have that, for any quantum relative entropy $D$, there is a finite inner and outer regular positive measure $\mu:\mc B(R)\to\R_+$ and $\lambda\geq0$ such that, for any states $\rho$ and $\sigma$ like in \eqref{eq:specialform},
$$
D(\rho\|\sigma)=-\int_R \log{\tr{\left(\sigma^{\frac{\alpha-1}{2z}}\rho^{\frac{\alpha}{z}}\sigma^{\frac{\alpha-1}{2z}}\right)^z}}\,d\mu(\alpha,z)-\lambda\log{\big\|({\rm supp}\,\rho)({\rm supp}\,\sigma)\big\|_\infty}
$$
where ${\rm supp}\,\tau$ is the support projection of a (positive) operator $\tau$ and $\|\cdot\|_\infty$ is the operator norm. The initial integral term above corresponds to a barycentre over the $\R_+^{\rm op}$-part of the test spectrum and the final term is the pointwise $z\to\infty$ limit $\hat{D}^\T$, the single $\T\R_+^{\rm op}$-quantum relative entropy of this semiring. See also Figure \ref{fig:alphazed} for the test spectrum of this semiring identified with the one-point compactification of $R$.
\end{example}

Next we review some facts that we may easily deduce for the quantum test spectrum $\hat{\mf D}(Q^d)$. Let us recall the notations and discussion of Example \ref{ex:classical}. The set of the equivalence classes of those $\vec{\rho}=\big(\rho^{(1)},\ldots,\rho^{(d)}\big)\in\mc Q^d$ such that $\rho^{(k)}\rho^{(\ell)}=\rho^{(\ell)}\rho^{(k)}$ for all $k,\ell=1,\ldots,d$, i.e.,\ {\it commutative tuples}, coincides exactly with the classical matrix majorization semiring $C^d$. Let us briefly elaborate this. For any commutative tuple $\vec{\rho}=\big(\rho^{(1)},\ldots,\rho^{(d)}\big)\in\mc Q^d_n$, we fix a common eigenbasis $\{h_i\}_{i=1}^n\subset\C^n$ for the operators $\rho^{(k)}$, and define non-normalized finite statistical experiment $P(\vec{\rho})=\big(p^{(1)}(\vec{\rho}),\ldots,p^{(d)}(\vec{\rho})\big)$ where $p^{(k)}(\vec{\rho})=\big(\sis{h_i}{\rho^{(k)}h_i}\big)_{i=1}^n$. Thus, $[P(\vec{\rho})]\in C^d$. We may obviously also define a classical-to-quantum preparation map that takes $P(\vec{\rho})$ back to $\vec{\rho}$. Thus, also each element of the quantum test spectrum $\hat{\mf D}(Q^d)$ reduces to exactly one element of the test spectrum of $C^d$. We can interpret this so that, for any $\Delta\in\hat{\mf D}(Q^d)$, there exists a unique $\Delta_{\rm cl}$ among the maps given in \eqref{eq:ClTemperate}, \eqref{eq:ClTropical}, and \eqref{eq:ClDeriv} such that $\Delta(\vec{\rho})=\Delta_{\rm cl}\big(P(\vec{\rho})\big)$ for all commutative tuples $\vec{\rho}$.

\begin{figure}
\begin{center}
\begin{overpic}[scale=0.5,unit=1mm]{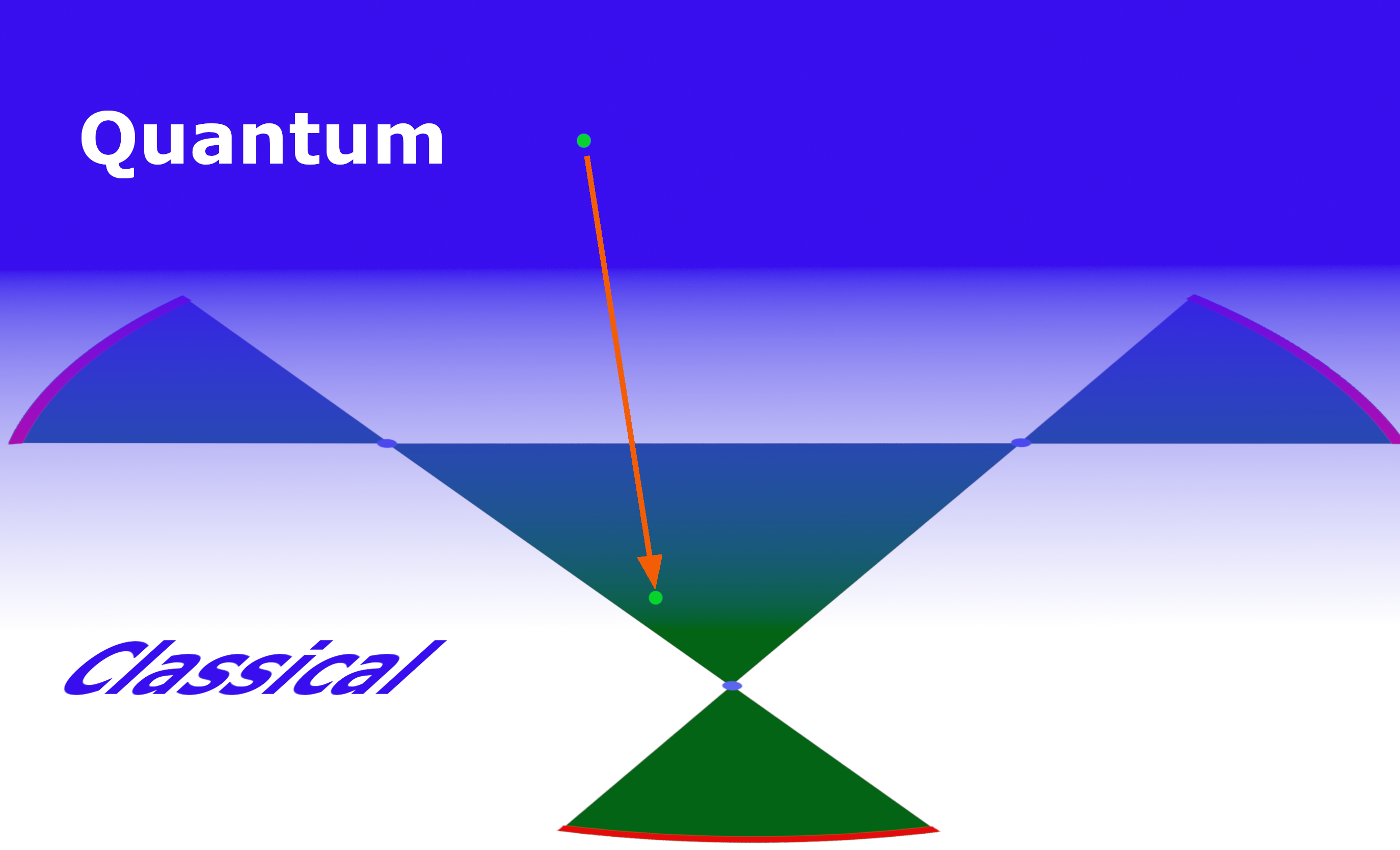}
\put(50,58){\begin{large}
$\textcolor{white}{\Delta\in\hat{\mf D}(Q^3)}$
\end{large}}
\put(55,20){\begin{large}
$\textcolor{white}{D_{\underline{\alpha}}}$
\end{large}}
\end{overpic}
\caption{\label{fig:quantum} An element of the quantum test spectrum $\hat{\mf D}(Q^d)$ (here a trivatiate divergence, $d=3$) reduces to a single classical divergence of the test spectrum $\hat{\mf D}(C^d)$ which in the depicted case is a temperate divergence $D_{\underline{\alpha}}$ as defined in \eqref{eq:ClTemperate}}
\end{center}
\end{figure}

In particular, for any $\Delta\in\mf D(Q^d,\R_+^{\rm op})$, there is a unique $\underline{\alpha}\in A_+$ such that $\Delta(\vec{\rho})=D_{\underline{\alpha}}\big(P(\vec{\rho})\big)$ for all commutative $\vec{\rho}$, i.e.,\ $\Delta\in\mf D(Q^d,\R_+^{\rm op})$ is associated to a particular $\underline{\alpha}\in A_+$. For any $\Delta\in\mf D(Q^d,\R_+)$, there exists unique $k\in\{1,\ldots,d\}$ and $\underline{\alpha}\in A_k$ such that $\Delta(\vec{\rho})=D_{\underline{\alpha}}\big(P(\vec{\rho})\big)$ for all commutative $\vec{\rho}$. Moreover, for any $\Delta\in\mf D(Q^d,\T\R_+)$, there exists unique $k\in\{1,\ldots,d\}$ and $\underline{\beta}\in B_k$ such that $\Delta(\vec{\rho})=D^\T_{\underline{\beta}}\big(P(\vec{\rho})\big)$ for all commutative $\vec{\rho}$. Finally, for any $k\in\{1,\ldots,d\}$ and any derivation $\Delta\in\mf D_k(Q^d)$, there exists a unique $\underline{\gamma}\in\R^d\setminus\{0\}$ such that $\Delta(\vec{\rho})=\Delta^{(k)}_{\underline{\gamma}}\big(P(\vec{\rho})\big)$ for all commutative $\vec{\rho}$. See also Figure \ref{fig:quantum}.

Some divergences presented earlier are immediately seen to be elements of the test spectrum as elaborated in the list below. For the non-derivation part of the spectrum (first three bullet points), all the divergences listed therein are extreme multivariate divergences (or relative entropies in the $d=2$ case) according to Theorem \ref{theor:extreme}.

\begin{itemize}
\item $\mf D(Q^d,\R_+^{\rm op})$: This set includes the matrix mean divergences $D^{\rm mm}_{\underline{\alpha}}$ of \eqref{eq:MatrixMeanDiv} whenever $\underline{\alpha}$ is a non-extreme probability vector. In the case $d=2$, these divergences reduce to the Kubo-Ando relative entropies $D^{\rm mm}_\alpha$ of \eqref{eq:KuboAndo} with $\alpha\in(0,1)$. Also in the case $d=2$, this part of the test spectrum includes the $\alpha$-$z$ relative entropies $D_{\alpha,z}$ of \eqref{eq:alphazed} whenever $\alpha\in(0,1)$ and conditions \eqref{eq:alphazedcond} hold. These include, in particular, the quantum R\'{e}nyi relative entropies of the Petz or sandwiched type with $\alpha\in(0,1)$. The classical counterparts of these are the divergences $D_{\underline{\alpha}}$ of \eqref{eq:ClTemperate} with $\underline{\alpha}$ a non-extreme probability vector. Also the staged divergences $D_{\alpha,z,\underline{\beta},(\underline{\beta}^a)_{a\in A}}$ of \eqref{eq:staged} are here whenever $\alpha\in[0,1)$.
\item $\mf D(Q^d,\R_+)$: In case $d=2$, this set includes the $\alpha$-$z$ relative entropies $D_{\alpha,z}$ of \eqref{eq:alphazed} whenever $\alpha>1$ and conditions \eqref{eq:alphazedcond} hold. These include, in particular, the quantum R\'{e}nyi relative entropies of the Petz and sandwiched type with $\alpha>1$. Also the Kubo-Ando relative entropies $D_{\alpha}^{\rm mm}$ for $1<\alpha\leq 2$ are here. The classical counterparts of these are the divergences $D_{\underline{\alpha}}$ of \eqref{eq:ClTemperate} with $\underline{\alpha}\in A_k$ for $k=1,\ldots,d$. Whenever $\alpha>1$, also the quantities $D_{\alpha,z,\underline{\beta},(\underline{\beta}^a)_{a\in A}}$ are here.
\item $\mf D(Q^d,\T\R_+)$: In the case $d=2$, this set includes only the max-quantum relative entropy $D_{\rm max}$ of \eqref{eq:maxdiv}. This corresponds to the classical max-relative entropy (the pointwise limit of the classical R\'{e}nyi divergences $D_\alpha$ as $\alpha\to\infty$).
\item  $\mf D_k(Q^d)$: In the case $d=2$ and $k=1$, this set includes both the Umegaki quantum relative entropy $D_{\rm U}$ of \eqref{eq:Umegaki} and the Belavkin-Staszewski quantum relative entropy$D_{\rm BS}$ of \eqref{eq:BS}. In the general case, we can, e.g.,\ define $\Delta^{(k)}_{\underline{\gamma},{\rm U}}\in\mf D_k(Q^d)$ for $\underline{\gamma}\in\R_+^d\setminus\{0\}$ through
$$
\Delta^{(k)}_{\underline{\gamma},{\rm U}}(\vec{\rho})=\frac{1}{\tr{\rho^{(k)}}}\sum_{\ell=1}^d \gamma_\ell D_{\rm U}\big(\rho^{(k)}\big\|\rho^{(\ell)}\big).
$$
We may define similar divergences w.r.t.\ the Belavkin-Staszewski divergence or mixing both Umegaki and Belavkin-Staszewski divergences. When $\underline{\gamma}$ is supported in more than one entry, these divergences are naturally not extreme in $\mf D_Q^d$.
\end{itemize}

\section{Conclusions}
We have defined extensive divergences on a preordered semiring and shown that, when the semiring $S$ is a preordered semidomain of polynomial growth and degeneracy $d\in\N$ (i.e.,\ the equivalence relation generated by the preorder is characterized by a vector-valued degenerate homomorphism $\|\cdot\|:S\to\R_{>0}^d\cup\{0\}$), then all the extensive divergences have a barycentric decomposition over the test spectrum of $S$. The test spectrum consists of extensive divergences that either satisfy certain extra additivity properties or Leibniz-type rules. We also show that the `bulk' of this test spectrum can be viewed as the set of convex extreme points of the convex set of extensive divergences.

All the above results are applicable to the problems of finite comparing classical or quantum experiments. Whereas we can fully characterize the test spectrum in the classical case, the quantum case poses problems that, for now, seem difficult to solve. However, we have shown that, once fully characterized, the quantum test spectrum is all that is needed to describe all extensive multivariate quantum divergences (or quantum relative entropies). However, most of the quantum divergences (relative entropies) proposed in the literature, such as the $\alpha-z$ relative entropies (including the Petz-type and sandwiched quantum R\'{e}nyi relative entropies, especially the Umegaki relative entropy) and the matrix mean relative entropies (Kubo-Ando relative entropies), are within this test spectrum and hence, due to our extremality result, essentially independent of each other. This strongly suggests that all of these varied quantum relative entropies are meaningful in some quantum information processing tasks, although we cannot say so definitively based solely on their extremality.

It is important to note that the real-algebraic methods and results are not so important for the barycentric result of Theorem \ref{thm:barycentre}. When perusing the proof of Theorem \ref{thm:barycentre}, the reader may notice that the only thing is needed is a variant of Lemma \ref{lemma:key} which identifies a set $\hat{\mf D}$ of extensive (when suitably defined) and monotone maps on the set $S_N$ of resources of interest which is compact in the topology of pointwise convergence such that, when $\Delta(x)\geq\Delta(y)$ for some $x,y\in S_N$, then $D(x)\geq D(y)$ for any extensive monotone map $D:S_N\to\R$. Proceeding in exactly the same way as in the proof of Theorem \ref{thm:barycentre}, one proofs similar barycentric results for extensive monotone maps on $S_N$.

One may also notice that we have proven a barycentric result first followed by an extremality result in Theorem \ref{theor:extreme}. This seems a bit odd since usually order is the opposite: we first identify the extreme points of a convex (and compact) set and then use a Choquet-like theorem to express any element of the convex set as a barycentre over the extreme points. It may well be that there is a more straight way to the extremality result Theorem \ref{theor:extreme} without using the barycentric result Theorem \ref{thm:barycentre}. From Theorem \ref{theor:extreme} one could then prove Theorem \ref{thm:barycentre} using Choquet integrals. We leave this possible more direct proof as an open problem for now.

\section*{Acknowledgements}
\noindent This project is supported by the National Research Foundation, Singapore through the National Quantum Office, hosted in A*STAR, under its Centre for Quantum Technologies Funding Initiative (S24Q2d0009). The author also thanks Frits Verhagen, Roberto Rubboli, Marco Tomamichel, Milan Mos\'{o}nyi, and P\'{e}ter Vrana for stimulating discussions when preparing this work.

\bibliographystyle{ultimate}
\bibliography{bibliography}

\end{document}